\documentclass[11pt]{article} 

\usepackage[utf8]{inputenc} 

\usepackage{geometry} 
\usepackage{graphicx} 
\usepackage{color}

\definecolor{Red}{rgb}{1,0,0}
\definecolor{Blue}{rgb}{0,0,1}
\definecolor{Olive}{rgb}{0.41,0.55,0.13}
\definecolor{Green}{rgb}{0,1,0}
\definecolor{MGreen}{rgb}{0,0.8,0}
\definecolor{DGreen}{rgb}{0,0.55,0}
\definecolor{Xellow}{rgb}{1,1,0}
\definecolor{Cyan}{rgb}{0,1,1}
\definecolor{Magenta}{rgb}{1,0,1}
\definecolor{Orange}{rgb}{1,.5,0}
\definecolor{Violet}{rgb}{.5,0,.5}
\definecolor{Purple}{rgb}{.75,0,.25}
\definecolor{Brown}{rgb}{.75,.5,.25}
\definecolor{Grey}{rgb}{.5,.5,.5}

\usepackage{booktabs} 
\usepackage{array} 
\usepackage{paralist} 
\usepackage{verbatim} 
\usepackage{subfigure} 
\usepackage{amsmath,amssymb,amsthm,mathrsfs}
\usepackage[colorlinks,linkcolor=blue,citecolor=blue]{hyperref}
\usepackage{blkarray}

\usepackage{tikz, pgf, extarrows,tabu}

\usepackage{fancyhdr} 
 \theoremstyle{plain}

\newtheorem{theorem}{Theorem}

\newtheorem{corollary}{Corollary}[section]
\newtheorem{claim}{Claim}[section]
\newtheorem{lemma}{Lemma}[section]

\newtheorem{theorem*}{Theorem}   
\newtheorem{lemma*}{Lemma} 
\newtheorem{corollary*}{Corollary} 
\newtheorem*{remark*}{Remark}

\newtheorem{remark}{Remark}[section]

\theoremstyle{definition}
\newtheorem{definition}{Definition}

\def\cF{{\cal F}}

\def\cK{{\cal K}}

\def\cN{{\cal N}}

\def\cP{{\cal P}}

\def\cU{{\cal U}}

\newcommand{\real}{\ensuremath{\mathbb{R}}}
\newcommand{\defn}{\ensuremath{:  =}}

\newcommand{\E}{\ensuremath{\mathbb{E}}}

\newcommand{\ind}{\ensuremath{\perp \!\!\! \perp}}
\newcommand{\rt}{\ensuremath{\rightarrow}}

\newcommand{\norm}[1]{\left\lVert#1\right\rVert}

\newcommand{\abs}[1]{\left\lvert#1\right\rvert}

\newcommand{\tilb}{\tilde{b}}
\newcommand{\tilSigma}{\tilde{\Sigma}}
\newcommand{\hatSigma}{\hat{\Sigma}}
\newcommand{\tilV}{\tilde{V}}
\newcommand{\hatV}{\hat{V}}
\newcommand{\mbbR}{\mathbb{R}}
\newcommand{\mbfr}{\mathbf{r}}
\newcommand{\mcN}{\mathcal{N}}

\newcommand{\mdim}{\mbox{dim}}

\newcommand{\mspan}{\mbox{span}}

\newif\ifvacomments
\vacommentstrue

\usepackage{authblk}

\title{Unifying the Brascamp-Lieb Inequality and the Entropy Power Inequality }

\author[ \hspace{-1ex}]{Venkat Anantharam\thanks{Department of Electrical Engineering and Computer Sciences, UC Berkeley. Email: ananth@eecs.berkeley.edu}}
\author[ \hspace{-1ex}]{Varun Jog \thanks{Department of Pure Mathematics and Mathematical Statistics, University of Cambridge. Email: vj270@cam.ac.uk}}
\author[ \hspace{-1ex}]{Chandra Nair\thanks{Department of Information Engineering Engineering, CUHK. Email: chandra@ie.cuhk.edu.hk}}

\affil[ ]{}
\date{}
\begin{document}
\maketitle
\vspace{-1.5cm}

\begin{abstract}
The entropy power inequality (EPI) and the Brascamp-Lieb inequality (BLI) are fundamental inequalities concerning the differential entropies of linear transformations of random vectors.
The EPI provides lower bounds for the differential entropy of linear transformations of random vectors with independent components. The BLI, on the other hand, provides upper bounds on the differential entropy of a random vector in terms of the differential entropies of some of its linear transformations. In this paper, we define a 
family of entropy functionals, which we show are subadditive. We then establish that Gaussians are extremal for these functionals by mimicking the idea in Geng and Nair (2014). As a consequence, we obtain a new entropy inequality that generalizes both the BLI and EPI. By considering a variety of independence relations among the components of the random vectors appearing in these functionals, we also obtain families of inequalities that lie between the EPI and the BLI.\footnote{A version of this paper appeared in the Proceedings of the IEEE International Symposium on Information Theory, 2019.}
\end{abstract}

\section{Introduction}
Information inequalities provide some of the most powerful mathematical tools in an information theorist's toolbox and are therefore a vital part of information theory. Inequalities such as the non-negativity of mutual information and the data processing inequality are so fundamental to information theory that they are inseparable from information-theoretic notation. These basic inequalities, combined with Fano's inequality, are powerful enough to yield the converse of Shannon's channel coding theorem. For harder problems in network information theory, it is necessary to develop more nuanced information inequalities. Not surprisingly, it is often the case that discovering new inequalities leads to  breakthroughs in network information theory problems. Some examples of information inequalities that spurred such breakthroughs include the entropy power inequality \cite{Sha48, Bla65}, numerous strengthened forms of the entropy power inequality \cite{Cos85, ZamFed93, Cou18}, strong data processing inequalities \cite{PolWu17}, and inequalities that established certain continuity properties of entropy \cite{PolWu16}. 

On a related note, ``single-letter characterizations" of a capacity region or outer bounds to a capacity region in network information theory are induced by subadditive functionals that reduce the characterization of the region to one governed by a single channel use. In this paper, we identify a new functional that is sub-additive and for which Gaussian distributions are extremal. Consequently, we obtain a new class of information inequalities that unifies two fundamental inequalities: the entropy power inequality (EPI) and the Brascamp-Lieb inequality (BLI). In what follows, we provide a brief introduction to the EPI and the BLI and state our main results.

As notational conventions in what follows, $:=$ and $=:$ denote equality by definition depending on whether the expression being defined is on the left or on the right respectively, while, for an integer $n > 0$, $[n]$ denotes $\{1, \ldots, n\}$ and
$I_{n \times n}$ denotes the $n \times n$ identity matrix.
We use the notation $|A|$ for the determinant of a square matrix $A$. We use the term ``entropy'' as synonymous with ``differential entropy'' in this document. All vectors are assumed to be column vectors, and we will adopt the convention that if $X$ is 
an
$\real^k$-valued vector and $Y$ is 
an
$\real^l$-valued vector, then 
$(X,Y)$ denotes the $\real^{k+l}$-valued vector that would normally be written as $(X^T, Y^T)^T$. Given a random vector $(Z_1, \ldots, Z_n)$, we use the notation $Z_{a:b}$ to denote the random vector $(Z_a, Z_{a+1}, \dots, Z_b)$, where $1 \le a \le b \le n$.
The notation $X \rt U \rt Y$ for random vectors $X$, $U$, and $Y$ indicates that $X$ and $Y$ are conditionally independent given $U$.

\paragraph{Entropy power inequality:} 
The EPI states that for any independent $\real^n$-valued random variables $X$ and $Y$, the following inequality holds:
\begin{equation}\label{eq: vanillaEPI}
e^{\frac{2h(X+Y)}{n}} \geq e^{\frac{2h(X)}{n}} + e^{\frac{2h(Y)}{n}}.
\end{equation}
Here, $h(\cdot)$ refers to the differential entropy function and all the differential entropies in equation \eqref{eq: vanillaEPI} are assumed to exist. Equality holds if and only if $X$ and $Y$ are Gaussian random variables with proportional covariance matrices. The EPI was proposed by Shannon \cite{Sha48} and was first proved by Stam \cite{Sta59}. This proof was later simplified by Blachman \cite{Bla65}. A variety of simple and ingenious proofs have been discovered since; see Rioul~\cite{Rio11} for a discussion. 

The EPI has an equivalent formulation due to Lieb~\cite{Lie02b} which is that for $\lambda \in (0,1)$ we have:

\begin{equation}\label{eq: liebform}
h(\sqrt \lambda X + \sqrt{1-\lambda} Y) \geq \lambda h(X) + (1-\lambda) h(Y).
\end{equation}
Equality holds in the above inequality if and only if $X$ and $Y$ are Gaussian random variables with identical covariance matrices. Note that $\sqrt \lambda X + \sqrt{1-\lambda}Y$ may be interpreted as a linear transformation of an $\real^{2n}$-valued random variable $Z \defn (X, Y)$ with some independence constraints on the components of $Z$, namely $X \ind Y$. Another result along such lines is Zamir and Feder's EPI~\cite{ZamFed93} for linear transformations of random vectors with independent components. This EPI has an equivalent formulation, discovered in \cite{Rio11, ZamFed93b}, that is analogous to the one in equation \eqref{eq: liebform}: For an $\real^n$-valued random vector $X \defn (X_1, \dots, X_n)$ with independent scalar components and any $k \times n$ matrix $A$ satisfying $AA^T = I_k$, we have
\begin{equation}\label{eq: zamfed}
h(AX) \geq \sum_{j=1}^n \alpha_j^2 h(X_j),
\end{equation}
where $\alpha_j^2$ is the squared-norm of the $j$-th column of $A$; i.e., $\alpha_j^2 := \sum_{i=1}^k a_{ij}^2$. 

\paragraph{Brascamp-Lieb inequality:} 

The BLI \cite{BraLie76} is actually a family of functional inequalities that lies, in some sense, at the intersection of information and functional inequalities. Many well-known and commonly used inequalities are special cases of the BLI, including H\"{o}lder's inequality, the Loomis-Whitney inequality, the Pr{\'e}kopa-Leindler inequality, and sharp forms of Young's convolution inequalities \cite{BenEtal08}. In Gardner's extensive survey \cite{Gar02}, the author describes relationships between popular functional and information inequalities using a pyramid-like sketch, where inequalities at the top imply those below. The BLI and its reverse lie at the very apex of this inequality pyramid. A simple statement of the BLI is as follows:
\begin{theorem}[Functional form of the BLI]\label{thm: BL1}
For $j \in [m]$, let $E$, $E_j$ be Euclidean spaces, $A_j: E \to E_j$ be linear maps, $c_j$ be positive real numbers, and $f_j$ be nonnegative integrable functions on $E_j$. Define the function $\cF$ via 
\begin{equation*}
\cF(f_1, \dots, f_m) := \frac{\int_{E} \prod_{j=1}^m f_j^{c_j}(A_jx)dx}{\prod_{j=1}^m \left( \int_{E_j} f_j(x_j) dx_j \right)^{c_j}}.
\end{equation*}
Then the supremum of $\cF$ over all  nonnegative and integrable $f_j$ is equal to the supremum of $\cF$ when $f_j$ are centered Gaussian functions; i.e., for all $j \in [m]$, we have $f_j(x_j) \propto e^{-x_j^T B_j x_j}$ for some positive semidefinite $B_j$.
\end{theorem}

Surprisingly, a direct connection exists between the functional form of the BLI and a generalized subadditivity result for entropy. This link was first discovered in Carlen, Lieb, and Loss~\cite{CarEtal04}, and has since led to newer proofs and generalizations of the original BLI \cite{CarEtal09, BarEtal06, CorLed10, Nai14, LiuEtal18}. The information-theoretic form of the BLI is the following:
\begin{theorem}[Information-theoretic form of the BLI,  Theorem 2.1 in Carlen and Cordero-Erausqin~\cite{CarEtal09}]\label{thm: BL2}
For $i \in [m]$, let $E$, $E_i$, $A_i$, and $c_i$ be as in Theorem \ref{thm: BL1}. For a random variable $X$ on $E$ with a well-defined differential entropy (see Definition \ref{def: h}) and
satisfying $E[\| X\|_2]^2 < \infty$, define $f(X)$ as
\begin{equation}\label{eq: BL2}
f(X) := h(X) - \sum_{j=1}^m c_j h(A_j X). 
\end{equation}
Then the supremum of $f$ over all such random variables $X$ is equal to the supremum of $f$ over all Gaussian random variables.
\end{theorem}
This information-theoretic form is completely equivalent to the functional form: For a fixed choice of the $a_j$ and the $c_j$, the supremums in both problems have a direct relationship and the cases of equality are also in correspondence \cite[Theorem 2.1]{CarEtal09}. A defining feature of the BLI is that it reduces an infinite-dimensional optimization problem to a finite-dimensional optimization problem over a set of positive definite matrices. When the supremum in Theorem \ref{thm: BL2} is finite, random variables that achieve the supremum are called \emph{extremizers}, and Gaussian random variables that achieve the supremum are called \emph{Gaussian extremizers}. 
\footnote{
In \cite{BenEtal08}
a Gaussian extremizer is defined as 
a distribution that extremizes among the class of Gaussian distributions, but it turns out that this definition is identical to the one used here.
}
The existence of extremizers or Gaussian extremizers and the finiteness of $D$ are not addressed by Theorem \ref{thm: BL2}, as stated above. However, this is well-understood in the literature \cite{Bar98, CarEtal09, BenEtal08}. 

\paragraph{Our contributions:}
The classical EPI and the EPI of Zamir and Feder are valid only under certain independence assumptions. To be precise, for an $\real^{2n}$-valued random vector $Z$, the EPI requires independence of $Z_{1:n}$ and $Z_{n+1:n}$ and considers the sum of these two vectors, whereas Zamir and Feder's EPI requires all the components to be independent and considers linear transformations of $Z$. It is natural to consider more general ``mixed" independence constraints, for instance, independence  of $Z_{1:k_1}, Z_{k_1+1:k_2}, \dots, Z_{k_r+1:n}$ for suitable choices of $k_i$, and establish lower bounds on $h(AZ)$ for a matrix $A$. This is indeed a special case of the setting considered in our work.

Consider an $\real^n$-valued random vector $X := (X_1, \dots, X_k)$, where $k \leq n$ and $X_i$ are mutually independent $\real^{r_i}$-valued random variables. Note that $\sum_{i=1}^k r_i = n$. We consider the following function:
\begin{equation}\label{eq: f_unified}
f(X) \defn \sum_{i=1}^k d_i h(X_i) - \sum_{j=1}^m c_j h(A_j X),
\end{equation}
for positive constants $d_i$ and $c_j$ where $i \in [k]$ and $j \in [m]$ for some $m \geq 1$, and surjective linear transformations $A_j$ from $\real^n$ to $\real^{n_j}$. Just as in Theorem \ref{thm: BL2}, our main result in Theorem \ref{thm: EPI+BL} states that the supremum of $f(\cdot)$ over all random variables $X$ satisfying the stated independence constraints is the same as the supremum evaluated over Gaussian random variables. In Theorem \ref{thm: finite}, we identify necessary and sufficient conditions on $n$, $k$, $m$
and the $r_i$, $d_i$, $c_j$, $n_j$ and $A_j$, such that this supremum is finite. We show that the EPI, BLI, and Zamir and Feder's EPI easily follow from Theorem \ref{thm: EPI+BL}. Theorem \ref{thm: EPI+BL} also provides a generalization of Zamir and Feder's result for certain kinds of dependent random variables.

Our main technical contribution is identifying new entropic functionals and proving that they satisfy a certain subadditivity property. The work of Geng and one of the authors~\cite{GengNair14} highlighted the critical role played by subadditivity in information inequalities. How subadditivity of information theoretic functionals---which is established using the chain rule and data processing relations---can be used to determine the capacity of the Gaussian vector broadcast channel was demonstrated in that work. Once subadditivity is ascertained, a technique from functional analysis called the ``doubling trick'' may be used to establish Gaussian optimality. The doubling trick, attributed to Ball \cite{Bar98b},  appeared in Lieb~\cite{Lie02} to prove that Gaussian kernels have Gaussian optimizers, and in Carlen~\cite{Car91} to show Gaussian optimality in the log-Sobolev inequality. Subadditivity followed by the doubling trick has been used to prove numerous information inequalities in recent years~\cite{CouJia14, KimEtal16, Gol16,  YangEtal17, ZhaEtal18, Cou18}.

\paragraph{Related work:} 

The EPI may be thought of as a limiting special case of the BLI. Gardner~\cite{Gar02} showed that the EPI follows from the sharp form of Young's inequality, which in turn is a special case of the BLI. This proof strategy is further clarified using a more geometric approach by Cordero-Erausquin and Ledoux \cite{CorLed10}. The authors of~\cite{CorLed10} establish the EPI directly from Theorem \ref{thm: BL2} by carefully choosing the $a_j$ and $c_j$ as a function of a parameter $\epsilon$ that tends to 0 and yields the EPI in the limit. While these are intriguing connections, they do not suggest concrete approaches for developing information inequalities for random vectors under more general independence constraints.

Various information-theoretic analogues of hypercontractive inequalities and reverse Brascamp-Lieb inequalities in finite alphabet spaces have been studied in \cite{AhlGac76, Nai14,BeiNai16}. A closely related work is that of Liu et al. \cite{LiuEtal18}, where a novel functional inequality called the forward-reverse Brascamp-Lieb inequality
is formulated, and it is shown that there exists an analogous information-theoretic version of this inequality. Most relevant to us is the forward-reverse Brascamp-Lieb inequality with linear maps that was introduced in Liu et al.~\cite{LiuEtal18}. Define a function $F$ of the marginal densities of an $\real^n$-valued random variable $X$:
\begin{equation}\label{eq: liu}
F(X_1, \dots, X_n) := \inf_{\{Y \mid Y_i \stackrel{d}= X_i,~ i \in [n]\}} \sum_{i=1}^n d_i h(Y_i) - \sum_{j=1}^m c_j h(A_j Y).
\end{equation}
Here, by $Y_i \stackrel{d}= X_i$ we mean that the distribution of $Y_i$ is identical to that of $X_i$. Theorem 8 in \cite{LiuEtal18} states that the supremum of $F$ is obtained when each $X_i$ is a centered Gaussian random variable, in which case the infimum in the definition in equation \eqref{eq: liu} is attained when the optimal coupling $Y$ is a jointly Gaussian random vector. The expressions in equations \eqref{eq: f_unified} and \eqref{eq: liu} look very similar. The main difference is that equation \eqref{eq: liu}  has an infimum over all possible couplings $Y$, whereas our definition in equation \eqref{eq: f_unified} enforces the unique coupling where the components $Y_i$ are mutually independent.\\

\paragraph{Structure of the paper: }

In Section \ref{section: prelim}, we introduce some preliminaries and set up the notation to be used in the rest of the paper. In Section \ref{section: main} we state our main result in Theorem \ref{thm: EPI+BL} and show that the EPI, BLI, and Zamir and Feder's EPI may be proved as special cases of this result.
In Section \ref{section: tp}, we prove Theorem \ref{thm: EPI+BL}. In Section \ref{section: finite}, we establish necessary and sufficient conditions for the supremum of $f$ in the expression in equation \eqref{eq: f_unified}  to be finite. In Section \ref{section: special}, we provide a concrete example that demonstrates the utility of Theorem \ref{thm: EPI+BL} in obtaining EPI-like results for dependent random variables. Finally, in Section \ref{section: end} we conclude the paper and describe some open problems.

\section{Preliminaries and notation}\label{section: prelim}
\begin{definition}\label{def: h}
For $n > 0$, let $X$ be an $\real^n$-valued random variable with density $f_X$ that lies in the convex set of probability densities
\begin{align}\label{eq: CarEtal_condition}
\left\{f ~\Big|~ \int_{\real^n} f(x) \log (1+f(x)) dx < \infty \right\}.
\end{align}
Then we define the entropy of $X$ as
\begin{equation}\label{eq: defh}
h(X) := -\int_{\real^n} f_X(x)\log f_X(x) dx.
\end{equation}
The entropy of a $0$-dimensional random variable is defined to be 0. 

\begin{remark}  \label{rem:inherit}
The integral in equation \eqref{eq: CarEtal_condition} is well-defined since the integrand is non-negative. The condition in equation \eqref{eq: CarEtal_condition} implies that the differential entropy integral 
in equation \eqref{eq: defh} is well-defined and lower-bounded away from $-\infty$. Also note that the condition in equation \eqref{eq: CarEtal_condition} is inherited by marginalization, i.e. if $f$ satisfies the condition and $g$ is a (multidimensional) marginal of $f$, then $g$ also satisfies the condition.
\end{remark}
\end{definition}
\begin{definition}
[BL datum]
For an integer $m > 0$, define an $m$-transformation as a triple 
$$\mathbf A \defn (n, \{n_j\}_{j \in [m]},  \{A_j\}_{j \in [m]}),$$
where for each $j \in [m]$, $A_j: \real^{n} \to \real^{n_j}$ is a surjective linear transformation, and $n_j \ge 0$.
An $m$-exponent is defined as an $m$-tuple $\mathbf c = \{c_j\}_{j \in [m]}$, such that $c_j \ge 0$ for $j \in [m]$. A Brascamp-Lieb datum (BL datum) is defined as a pair $(\mathbf A, \mathbf c)$ where $\mathbf A$ is an $m$-transformation and $\mathbf c$ is an $m$-exponent, for an integer $m > 0$. 
\end{definition}

\begin{definition}[EPI datum]
For an integer $k > 0$, define a $k$-partition of $n$ as $\mathbf r = \{r_i\}_{i \in [k]},$
such that $r_i > 0$ are integers and $\sum_{i \in [k]} r_i = n$. 
Let $\mathbf d = \{d_i\}_{i \in [k]}$ such that $d_i \geq 0$ for all $i$
be a $k$-exponent. An EPI datum is a pair $(\mathbf r, \mathbf d)$ where $\mathbf r$ is a $k$-partition and $\mathbf d$ is a $k$-exponent, for an integer $k > 0$.
\end{definition}

\begin{definition}[BL-EPI datum]
For an integer $n > 0$, a BL-EPI datum is defined as $(\mathbf A, \mathbf c, \mathbf r, \mathbf d)$ where $(\mathbf A, \mathbf c)$ is a BL datum for an integer $m > 0$, and $(\mathbf r, \mathbf d)$ is an EPI datum for an integer $k > 0$. 
\end{definition}

\begin{definition}\label{def: pg}
Let $(\mathbf A, \mathbf c, \mathbf r, \mathbf d)$ be a BL-EPI datum where $\mathbf r$ is a $k$-partition of $n$. Define $\cP(\mathbf r)$ to be the set of all $\real^n$-valued random vectors $X \defn (X_1, X_2, \dots, X_k)$ such that:
\begin{enumerate}
\item
For $i \in [k]$, the random vectors $X_i$ take values in $\real^{r_i}$ and their densities satisfy the condition in equation \eqref{eq: CarEtal_condition};
\item
$X_1, X_2, \dots, X_k$ are independent;
\item
$\E X = 0$ and $\E \norm{X}_2^2 < \infty$.
\end{enumerate}
Since entropy expressions are not affected by adding constants, the 0-mean assumption in Definition \ref{def: pg} may be made without loss of generality. Define $\cP_g(\mathbf r) \subseteq \cP(\mathbf r)$ as the set of random variables $X$ that satisfy the properties above, while, in addition, each $X_i$, $i \in [k]$ is Gaussian.
\end{definition} 

\begin{remark}
Whether an $\real^{n}$-valued random vector $X$ lies in $\cP(\mathbf r)$ or not is a property of its distribution. The finite variance assumption on random variables in $\cP(\mathbf r)$ implies that the entropies $h(X_i)$ for $i \in [k]$ and $h(A_jX)$ for $j \in [m]$ are bounded away from $\infty$. However, with only the  variance assumption in place, it may happen that some of these entropies equal $-\infty$, which happens, for instance, when $X$ is a constant. In this paper, we shall be dealing with differences of entropies of the form
\begin{align}\label{eq: diff_h}
    \sum_{i=1}^k d_i h(X_i) - \sum_{j=1}^m c_j h(A_j X).
\end{align}
The condition in equation
\eqref{eq: CarEtal_condition} together with the finite variance assumption has the effect of ensuring that the absolute values of the differential entropies are finite, which ensures that the above difference is well-defined for $X \in \cP(\mathbf r)$. This is a technical assumption made for ease of presentation. In cases where the expression in equation \eqref{eq: diff_h} is not well-defined, we may redefine it to equal the limit
\begin{equation}\label{eq: fM1}
     \limsup_{\delta \to 0_+} \sum_{i=1}^k d_i h(\tilde X_i) - \sum_{j=1}^m c_j h(A_j \tilde X + \sqrt \delta Z_j),  
\end{equation}
where  $\tilde X := X + \sqrt \delta W$ for a standard normal $W$ independent of $X$ and the $Z_j$ are standard normal random vectors independent of $(X,W)$.
With this modification, our results continue to hold for random variables that satisfy 
all the conditions in 
Definition \ref{def: pg} except 
the condition in equation
\eqref{eq: CarEtal_condition}.

\end{remark}

The following two concepts are required for Theorem \ref{thm: finite}.

\begin{definition}\label{def: r_product}
Let $(\mathbf A, \mathbf c, \mathbf r, \mathbf d)$ be a BL-EPI datum. Define a subspace $V \subseteq \real^n$ as being of $\mathbf r$-product form if $V$ may be written as  $V = V_1 \times V_2 \times \dots \times V_k$ for subspaces $V_i \subseteq \real^{r_i}$, for $i \in [k]$.
\end{definition}

\begin{definition}\label{def: critical_V}
Let $(\mathbf A, \mathbf c, \mathbf r, \mathbf d)$ be a BL-EPI datum. An $\mathbf r$-product form subspace $V \subseteq \real^n$ is called a \emph{critical subspace} if 
\begin{equation*}
\sum_{i=1}^k d_i \dim(V_i)  = \sum_{j=1}^m c_j \dim(A_j V).
\end{equation*}
\end{definition}

\begin{definition} \label{def:M_notation}
For a BL-EPI datum $(\mathbf A, \mathbf c, \mathbf r, \mathbf d)$, define $M(\mathbf A, \mathbf c, \mathbf r, \mathbf d)$ as
\begin{align*}
M(\mathbf A, \mathbf c, \mathbf r, \mathbf d) := \sup_{X \in \cP(\mathbf r)} \sum_{i=1}^k d_i h(X_i) - \sum_{j-1}^m c_j h(A_j X).
\end{align*}
Similarly, define $M_g(\mathbf A, \mathbf c, \mathbf r, \mathbf d)$ as the above supremum taken over Gaussian inputs $X \in \cP_g(\mathbf r)$. When the BL-EPI datum is fixed, we shall omit the $(\mathbf A, \mathbf c, \mathbf r, \mathbf d)$ argument and use the simplified notation $M$ and $M_g$.
\end{definition}

\section{Main results}\label{section: main}
We are now in a position to state our main result:
\begin{theorem}[Unified EPI and BLI]\label{thm: EPI+BL}
Let $(\mathbf A, \mathbf c, \mathbf r, \mathbf d)$ be a BL-EPI datum. Recall the definition
\begin{equation}
M_g := \sup_{Z \in \cP_g(\mathbf r)}\sum_{i=1}^k d_i h(Z_{i}) - \sum_{j=1}^m c_j h(A_j Z).
\end{equation}
Then for any $X \in \cP(\mathbf r)$, the following inequality holds:
\begin{equation}\label{eq: mg}
\sum_{i=1}^k d_i h(X_i) - \sum_{j=1}^m c_j h(A_j X) \leq M_g.
\end{equation}
\end{theorem}
Recall that in Definition \ref{def:M_notation} we introduced the quantity (with a simplified notation):
\begin{equation}\label{eq: m}
M := \sup_{X \in \cP(\mathbf r)} \sum_{i=1}^k d_i h(X_{i}) - \sum_{j=1}^m c_j h(A_j X).
\end{equation}
Naturally, we have $M \geq M_g$. Thus, if $M_g$ is $+\infty$, then so is $M$. If $M_g < \infty$, then the above result implies $M \leq M_g$, and thus $M = M_g$.  An equivalent way of stating the above result is asserting $M = M_g$. Theorem~\ref{thm: EPI+BL} does not address the following points, which are worth investigating:

\begin{enumerate}
\item
\textbf{Finiteness:}  When is $M_g$ (and therefore $M$) finite?
\item
\textbf{Extremizability and Gaussian extremizability:} Assuming $M$ is finite, when do extremizers exist for the supremum in equation \eqref{eq: m}, and when do Gaussian extremizers exist for the supremum in equation \eqref{eq: mg}? In particular, does extremizability imply Gaussian extremizability? (Clearly, the reverse implication is true because of Theorem
\ref{thm: EPI+BL}.)
\item
\textbf{Uniqueness of extremizers:} Assuming extremizers exists, are they unique in some appropriate sense? 
\end{enumerate}
The answers to all these questions will depend on the BL-EPI datum $(\mathbf A, \mathbf c, \mathbf r, \mathbf d)$. In this paper, we resolve the first question by identifying necessary and sufficient conditions on  $(\mathbf A, \mathbf c, \mathbf r, \mathbf d)$ that ensure finiteness of $M$ and $M_g$. We do not address the latter two questions here. We show the following result:

\begin{theorem}\label{thm: finite}
 For a BL-EPI datum $(\mathbf A, \mathbf c, \mathbf r, \mathbf d)$, we have $M(\mathbf A, \mathbf c, \mathbf r, \mathbf d) < \infty$ if and only if the following conditions are satisfied:

\begin{align}
&\sum_{i=1}^k d_i \dim(V_i) \leq \sum_{j=1}^m c_j \dim(A_j V) \quad \text{ for all $\mathbf r$-product form}~ V,  \quad \text{ and }  \label{eq: finite1} \\
&\sum_{i=1}^k d_i r_i = \sum_{j=1}^{m} c_j n_j. \label{eq: finite2}
\end{align}
\end{theorem}

As we show below, Theorem \ref{thm: EPI+BL} readily implies the EPI, BLI, and Zamir and Feder's EPI. For this reason, we choose to interpret the inequality in Theorem \ref{thm: EPI+BL} as a unified version of the Brascamp-Lieb inequality and the entropy power inequality. 

\paragraph{Entropy Power Inequality: }

We will prove the EPI  in Lieb's form \eqref{eq: liebform} using Theorem \ref{thm: EPI+BL}. Let $X$ and $Y$ be independent $\real^d$-valued random variables with zero means and bounded variances, and let $\lambda \in (0,1)$. The expression $\lambda h(X) + (1-\lambda) h(Y) - h(\sqrt \lambda X+ \sqrt{1-\lambda} Y)$ corresponds to $n = 2d$, $k=2$, $r_1=r_2=d$, $d_1 = \lambda$, $d_2 = 1-\lambda$, $c_1 = 1$, and $A_1 = [\sqrt \lambda I_d, \sqrt{1-\lambda} I_d]$. Note that it is enough to prove $M_g = 0$ by explicit calculation. Consider Gaussian random variables $Z_1 \sim \cN(0, \Sigma_1)$ and $Z_2 \sim \cN(0, \Sigma_2)$. Plugging in the entropies of these Gaussian random variables and simplifying, we see that we need to evaluate the supremum
\begin{align*}
    M_g = \sup_{\Sigma_1, \Sigma_2 \succeq 0} \lambda \log \det(\Sigma_1) + (1-\lambda) \log \det \Sigma_2 - \log \det(\lambda \Sigma_1 + (1-\lambda)\Sigma_2).
\end{align*}
This supremum is seen to be 0 via the concavity of the $\log\det$ function.

\paragraph{Brascamp-Lieb Inequality: }

When $k=1$, $r_1 = n$, and $d_1 = 1$,  we recover the setting of the Brascamp-Lieb inequality in its equivalent form of subadditivity of entropy:
\begin{equation}
h(X) \leq \sum_{j=1}^m c_j h(A_j X) + M_g,
\end{equation}
for all $\real^{n}$-valued random variables $X$ with $\E X = 0$
and $\E ||X||_2^2 < \infty$.

\paragraph{Zamir and Feder's Inequality: }
Let $A$ be a $k \times n$ matrix satisfying $AA^T = I_{k \times k}$.  For $1 \leq j \leq n$, let the squared norm of the $j$-th column of $A$ be denoted by $\alpha_j^2$; i.e.,
$$\alpha_j^2 := 
\sum_{i=1}^k a_{ij}^2.$$ 
Just as we did for the EPI, it is enough to show that $M_g \le 0$ by explicitly computing the supremum of $ \sum_{j=1}^n \alpha_j^2 h(X_j) - h(AX)$ over Gaussian $X$. Let $\Lambda = \operatorname{Diag}(\lambda_1, \lambda_2, \dotsc, \lambda_n)$ be a positive definite matrix. Define a function $F$ from the space of positive definite diagonal matrices to $\mathbb R$ as follows:
$$ F(\Lambda) = \log \lvert A \Lambda A^T\rvert-\sum_{j=1}^n \alpha_j^2 \log \lambda_j.$$

If we show that $F(\Lambda) \geq 0$, then Theorem \ref{thm: EPI+BL} will immediately imply Zamir and Feder's EPI for random vectors with independent components. Let $B := A \Lambda^{1/2}$, so that $A\Lambda A^T = BB^T$. Using the Cauchy-Binet formula for the determinant of $BB^T$, we obtain
$$|BB^T| = \sum_{1 \leq i_1 < \dots < i_k \leq n} |B_{i_1 i_2 \dots i_k}|| B_{i_1 i_2 \dots i_k}^T|, $$ 
where $B_{i_1 i_2 \dots i_k}$ consists of the $k$ columns of $B$ corresponding to the indices $i_1, \dots, i_k$. The right hand side of the above equality may be written explicitly as
$$ \sum_{1 \leq i_1 < \dots < i_k \leq n} \left({\prod_{j=1}^k \lambda_{i_j}}\right) |A_{i_1i_2\dots i_k}|^2.$$
Noting that $\sum_{1 \leq i_1 < \dots < i_k \leq n}|A_{i_1i_2\dots i_k}|^2 = |AA^T| = |I_k|= 1 $ (again via the Cauchy-Binet formula), we may take logarithms and use Jensen's inequality to obtain
$$ \log |A \Lambda A^T| \geq \sum_{1 \leq i_1 < \dots < i_k \leq n} |A_{i_1i_2\dots i_k}|^2 \log \left({\prod_{j=1}^k \lambda_{i_j}}\right). $$
We now gather  the coefficients of $\log {\lambda_j}$ for a fixed $j$. 
The coefficient of $\log {\lambda_1}$ is given by

$$ \sum_{1 = i_1 < \dots < i_k \leq n} |A_{i_1i_2\dots i_k}|^2= 1 - |A_{2,3,\dots, n}A_{2,3,\dots,n}^T| = 1 - |I_n - A_1 A_1^T| = \alpha_1^2.$$

Here, the first equality follows by using the Cauchy-Binet formula again, the second equality follows from the orthogonality of the rows of $A$, and the third equality is true because $|I_n-uu^T| = 1- \norm{u}^2$ for any vector $u$. A similar calculation can be done to show that the coefficient of $\log {\lambda_j}$ is $\alpha_j^2$
for all $1 \le j \le n$, which completes the proof of $F(\Lambda) \geq 0$.

\section{Proof of Theorem \ref{thm: EPI+BL}}\label{section: tp}

Our proof strategy relies on the technique of Geng and Nair~\cite{GengNair14} which was developed to solve optimization problems of the form $\sup_{\text{Cov}(X) \preceq \Sigma} s(X).$ A rough sketch of this proof strategy is outlined below:
\begin{itemize}
\item
\textbf{Concave envelope:} Define the concave envelope of $s$, denoted by $S$, as
the smallest concave function that pointwise dominates $s$.
It can be seen that 
\begin{align*}
S(X) = \sup_{U} s(X|U) = \sup_U \sum_{u \in \cU} s(X | U=u) p_U(u),
\end{align*}
where the supremum is over finite auxiliary random variables $U$ with support $\cU$. 
\item 
\textbf{Subadditivity of $S$:} This step consists of defining $S$ on the larger space of pairs of random variables $(X_1,X_2)$. A straightforward extension often exists for information-theoretic functions $S$. The subadditivity result shows that
\begin{align*}
S(X_1, X_2) \leq S(X_1) + S(X_2).
\end{align*}
The ingredients for establishing the subadditivity result developed in this paper stems from the ideas to establish converses to coding theorems and outer bounds in network information theory. An argument with a flavor similar to 
that employed here can be found outlined in \cite{Nai14b}.
\item
\textbf{Optimizers of $S$:} In this step (also known as the doubling trick), we consider two i.i.d.\ copies of any optimizer $X$ of $S(X)$, say $(X_1, X_2)$, and show that $(X_1+X_2)/\sqrt 2$ and $(X_1-X_2)/\sqrt 2$ are also optimizers of $S(X)$. 
From here, we may use Gaussian characterization results \cite{GhuOlk62} or the central limit theorem \cite{GengNair14} to conclude that it is enough to consider only Gaussian optimizers.
\item 
\textbf{Optimizers of $s$:} In this final step, we show that the optimal value 
for $S(X)$ 
is attained by a single Gaussian distribution; i.e., we may assume without loss of generality that $\abs\cU = 1$, 
and thus this Gaussian also maximizes 
$s(X)$.
\end{itemize}

The crux of the proof is establishing the subadditivity of $S$. Our proof relies on the expanding the joint entropy $h(X_1, X_2)$ in two separate ways as follows:
\begin{itemize}
\item[(A)]
$h(X_1, X_2) = h(X_1) + h(X_2) - I(X_1; X_2)$,
\item[(B)]
$h(X_1, X_2) = h(X_1|X_2) + h(X_2|X_1) + I(X_1;X_2)$.
\end{itemize}
To highlight the main ideas, we present a proof sketch of the subadditivity result for the EPI using our new technique.
\subsection{Proving the EPI via subadditivity} Consider the function
\begin{equation}
s(X_1, Y_1) \defn \lambda h(X_1) + (1-\lambda)h(Y_1) - h(\sqrt \lambda X_1 + \sqrt {1-\lambda} Y_1),
\end{equation}
where $X_1 \ind Y_1$. Define the lifting of $s$ to the space of pairs of random variables by
\begin{equation}\label{eq: s_pair_epi}
s(X_{1:2}, Y_{1:2}) \defn \lambda h(X_{1:2}) + (1-\lambda)h(Y_{1:2}) - h(\sqrt \lambda X_{1:2} + \sqrt {1-\lambda}Y_{1:2}),
\end{equation}
where $X_{1:2} \ind Y_{1:2}$. Let $S(X_1, Y_1)$ and $S(X_{1:2}, Y_{1:2})$ be the respective concave envelopes of $s$ and its lifting.
\footnote{  \label{fn:hull}
To get $S(X_1,Y_1)$ from $s(X_1, Y_1)$, we can think of the domain of $s(X_1, Y_1)$ as being the product of the convex set of probability densities
on $x$ satisfying \eqref{eq: CarEtal_condition} and the convex set of probability densities on $y$ satisfying 
\eqref{eq: CarEtal_condition}, and take the concave hull on this 
product space;
similarly for getting $S(X_{1:2}, Y_{1:2})$ from $s(X_{1:2}, Y_{1:2})$.
It can be checked that any product distribution on $(X_1,Y_1)$ got by a mixture of product distributions
can be viewed as having the mixing done on the marginals,
basically because if $p(x) q(y) = \sum_i \lambda_i p_i(x) q_i(y)$ where $\lambda_i \ge 0$ and $\sum_i \lambda_i =1$
then summing over $y$ on both sides gives
$p(x) = \sum_i \lambda_i p_i(x)$ 
and similarly $q(y) = \sum_i \lambda_i q_i(y)$.
This justifies why we can write \eqref{eq:mixing}
and the analogous expression for $S(X_{1:2}, Y_{1:2})$.

}
We would like to show the subadditivity relation
\begin{equation}
S(X_{1:2}, Y_{1:2}) \leq S(X_1, Y_1) + S(X_2, Y_2).
\end{equation}
Notice that 
\begin{equation}        \label{eq:mixing}
S(X_1, Y_1) = \sup_{X_1 \to U \to Y_1} s(X_1, Y_1|U),
\end{equation}
and similarly for $S(X_{1:2}, Y_{1:2})$. For any auxiliary random variable $U$ satisfying $X_{1:2} \rt U \rt Y_{1:2}$, applying expansion (A) to each entropy term in equation \eqref{eq: s_pair_epi} (conditioned on $U$) yields 
\begin{align}
s(X_{1:2}, Y_{1:2} \mid U) &= \lambda h(X_{1:2}|U) + (1-\lambda)h(Y_{1:2}|U)- h(\sqrt \lambda X_{1:2}  + \sqrt {1-\lambda} Y_{1:2} | U) \nonumber \\
&= \Big[\lambda h(X_1|U) + (1-\lambda)h(Y_1|U)- h(\sqrt \lambda X_1  + \sqrt {1-\lambda} Y_1 | U) \Big] \nonumber\\
& + \Big[\lambda h(X_2|U) + (1-\lambda)h(Y_2|U)- h(\sqrt \lambda X_2 + \sqrt {1-\lambda} Y_2 | U) \Big] \nonumber\\
& + \Big[- \lambda I(X_1; X_2 | U) - (1-\lambda) I(Y_1; Y_2 | U) + I(\sqrt \lambda X_1 + \sqrt{1-\lambda} Y_1; \sqrt \lambda X_2 + \sqrt{1-\lambda} Y_2| U)\Big]. \label{eq: epi_exp1}
\end{align}
For simplicity, call the terms in the brackets $T_1(U)$, $T_2(U)$, and $T_3(U)$ respectively, even though they actually depend on $p_{U|X_{1:2}, Y_{1:2}}$.  Observing that   $X_i \to U \to Y_i$ for $i= 1,2$, we may conclude $T_1(U) \leq S(X_1, Y_1)$ and $T_2(U) \leq S(X_2, Y_2)$. Substituting these inequalities, we arrive at 
\begin{align}
s(X_{1:2}, Y_{1:2}|U) \leq S(X_1, Y_1) + S(X_2, Y_2) + T_3(U).\label{eq: epi_f2}
\end{align}

We now expand the
expression in equation \eqref{eq: s_pair_epi} (conditioned on $U$) using expansion (B) for each entropy term:
\begin{align}
s(X_{1:2}, Y_{1:2}|U) &= \lambda h(X_{1:2}|U) + (1-\lambda)h(Y_{1:2}|U)- h(\sqrt \lambda X_{1:2}  + \sqrt {1-\lambda} Y_{1:2} | U) \nonumber\\
&= \Big[\lambda h(X_1|U, X_2) + (1-\lambda)h(Y_1|U, Y_2)- h(\sqrt \lambda X_1  + \sqrt {1-\lambda} Y_1 | U, \sqrt \lambda X_2 + \sqrt {1-\lambda} Y_2) \Big] \nonumber\\
& + \Big[\lambda h(X_2|U, X_1) + (1-\lambda)h(Y_2|U, Y_1)- h(\sqrt \lambda X_2 + \sqrt {1-\lambda} Y_2 | U, \sqrt \lambda X_1  + \sqrt {1-\lambda} Y_1) \Big] \nonumber\\
& + \Big[\lambda I(X_1; X_2 | U) + (1-\lambda) I(Y_1; Y_2 | U) - I(\sqrt \lambda X_1 + \sqrt{1-\lambda} Y_1; \sqrt \lambda X_2 + \sqrt{1-\lambda} Y_2| U)\Big]. \label{eq: epi_exp2}
\end{align}
For ease of notation, call the three terms $R_1(U)$, $R_2(U)$, and $R_3(U) = -T_3(U)$, even though they actually depend on $p_{U|X_{1:2}, Y_{1:2}}$. Similar to inequality \eqref{eq: epi_f2}, we would like to upper bound $R_1(U)$ and $R_2(U)$ by $S(X_1, Y_1)$ and $S(X_2, Y_2)$ respectively. However, the conditioning for the entropy terms in each of the $R_i(U)$ is not the same so we cannot directly conclude such a bound. Using the chain rule of mutual information and data-processing relations, we may make the conditioning in $R_1(U)$ and $R_2(U)$ uniform by introducing some extra mutual information terms:
\begin{align*}
R_1(U) &= \Big[\lambda h(X_1|U, X_2) + (1-\lambda)h(Y_1|U, Y_2)- h(\sqrt \lambda X_1  + \sqrt {1-\lambda} Y_1 | U, \sqrt \lambda X_2 + \sqrt {1-\lambda} Y_2) \Big]\\
&= \Big[\lambda h(X_1|U, X_2, Y_2) + (1-\lambda)h(Y_1|U, Y_2, X_2)- h(\sqrt \lambda X_1  + \sqrt {1-\lambda} Y_1 | U, X_2, Y_2) \Big]\\
&- I( \sqrt \lambda X_1  + \sqrt {1-\lambda} Y_1;  X_2, Y_2  | U, \sqrt \lambda X_2  + \sqrt {1-\lambda} Y_2)\\
&=: \tilde R_1(U) - I_1(U),
\end{align*}
where the notational conventions
$\tilde R_1(U)$ and $I_1(U)$ are used 
even though the respective terms actually depend on $p_{U|X_{1:2}, Y_{1:2}}$.
The main step in the preceding equation is justified as follows. First, it it easy to check using the Markov relation $(X_1, X_2) \rt U \rt (Y_1, Y_2)$ that
\begin{equation*}
h(X_1 | U, X_2) = h(X_1|U, X_2, Y_2), \quad \text{and} \quad h(Y_1 | U, Y_2) = h(Y_1 | U, X_2, Y_2).
\end{equation*}
Also, we may verify that 
\begin{align*}
&h(\sqrt \lambda X_1  + \sqrt {1-\lambda} Y_1 | U, \sqrt \lambda X_2 + \sqrt {1-\lambda} Y_2)\\
& = h(\sqrt \lambda X_1  + \sqrt {1-\lambda} Y_1 | U, X_2, Y_2) + I( \sqrt \lambda X_1  + \sqrt {1-\lambda} Y_1; X_2, Y_2  |U,  \sqrt \lambda X_2  + \sqrt {1-\lambda} Y_2).
\end{align*}
Similar reasoning for $R_2(U)$ gives 
\begin{align*}
R_2(U) &= \Big[\lambda h(X_2|U, X_1) + (1-\lambda)h(Y_2|U, Y_1)- h(\sqrt \lambda X_2  + \sqrt {1-\lambda} Y_2 | U, \sqrt \lambda X_1 + \sqrt {1-\lambda} Y_1) \Big]\\
&= \Big[\lambda h(X_2|U, X_1, Y_1) + (1-\lambda)h(Y_2|U, Y_1, X_1)- h(\sqrt \lambda X_2  + \sqrt {1-\lambda} Y_2 | U, X_1, Y_1) \Big]\\
&- I( \sqrt \lambda X_2  + \sqrt {1-\lambda} Y_2;X_1, Y_1  | U, \sqrt \lambda X_1  + \sqrt {1-\lambda} Y_1)\\
&=: \tilde R_2(U) - I_2(U),
\end{align*}
where the notational conventions
$\tilde R_2(U)$ and $I_2(U)$ are used 
even though the respective terms actually depend on $p_{U|X_{1:2}, Y_{1:2}}$.
Substituting the expressions for $R_1(U)$ and $R_2(U)$ in the expansion in equation \eqref{eq: epi_exp2}, we arrive at 
\begin{align}
s(X_{1:2}, Y_{1:2}|U)&= \tilde R_1(U) + \tilde R_2(U) - T_3(U) - I_1(U) - I_2(U) \nonumber \\
&\stackrel{(a)}\leq S(X_1, Y_1) + S(X_2, Y_2) - T_3(U) - I_1(U) - I_2(U) \nonumber\\
&\stackrel{(b)}\leq S(X_1, Y_1) + S(X_2, Y_2) - T_3(U). \label{eq: epi_f3}
\end{align}
Here, in step $(a)$ we used the Markov chains $X_1 \rt (U, X_2, Y_2) \rt Y_1$ and $X_2 \rt (U, X_1, Y_1) \rt Y_2$. Step $(b)$ follows by noticing that $I_1(U)$ and $I_2(U)$ are non-negative, being mutual information expressions. 

Inequalities \eqref{eq: epi_f2} and \eqref{eq: epi_f3} may now be used in tandem to conclude 
\begin{equation}\label{eq: epi_su1}
s(X_{1:2}, Y_{1:2} | U) \leq S(X_1, Y_1) + S(X_2, Y_2).
\end{equation}
Taking the supremum over all auxiliary random variables $U$ satisfying $X_{1:2} \rt U \rt Y_{1:2}$
leads to
\begin{equation}\label{eq: epi_su2}
S(X_{1:2}, Y_{1:2}) \leq S(X_1, Y_1) + S(X_2, Y_2).
\end{equation}
Notice that the above proof not only gives us subadditivity, but also states that if there is equality in equation \eqref{eq: epi_su1} for some optimal $U^*$, then
$I_1(U^*) = I_2(U^*) = T_3(U^*) = 0$. This leads to several independence conditions that can be used establish Gaussian optimality. We do not sketch this part of the proof here.

In what follows, we develop this outline into a rigorous proof for a more general result in two stages. In Section \ref{section: tensorization} we establish the key subadditivity inequality and the independence relations that follow from the conditions for equality in that inequality, and in Section \ref{section: proof} we complete the proof of Theorem \ref{thm: EPI+BL} by proving Gaussian optimality.

\subsection{Subadditivity lemma}\label{section: tensorization}
\subsubsection{Preliminaries}\label{section: tensorization: preliminaries}
Let $(\mathbf A, \mathbf c, \mathbf r, \mathbf d)$ be a BL-EPI datum. Let $X \defn (X_1, X_2, \dots, X_k) \in \cP(\mathbf r)$, where $X_i \sim p_{X_i}$. A natural definition for $s(X)$ would be
\begin{align*}
s(X) &:= \sum_{i=1}^k d_i h(X_i) - \sum_{j=1}^m c_j h(A_jX),
\end{align*}
and one might then work with its concave envelope $S$.
However, for technical reasons we consider Gaussian-smoothed random variables in defining $s$ as follows:

\begin{definition}
Let $W_i \sim \cN(0, I_{r_i \times r_i}), i \in [k]$ be mutually independent standard normal random variables on $\real^{r_i}$, and let $W:= (W_1, \ldots, W_k)$. For $j \in [m]$, define independent Gaussian random variables $Z_j \sim \cN(0, I_{n_j \times n_j})$, and let $Z := (Z_1, Z_2, \dots, Z_m)$.
Assume that the random variables $X$, $W$ and $Z$ are mutually independent. For $\epsilon, \delta \geq 0$ define $s_{\epsilon, \delta}: \cP(\mathbf r) \to \real$ as
\begin{equation}\label{eq: sdef}
s_{\epsilon, \delta}(X) \defn \sum_{i=1}^k d_i h(X_i + \sqrt \delta W_i) - \sum_{j =1}^m c_j h(A_j (X+\sqrt \delta W) + \sqrt \epsilon Z_j).
\end{equation}
\end{definition}

Let $S_{\epsilon, \delta}$ be the concave envelope of $s_{\epsilon, \delta}$.
Let $U$ be an auxiliary random variable taking values in a finite set $\cU$ such that we have
$p_{X|U}(\cdot|U) \in \cP(\mathbf r)$. It is easy to see that the concave envelope has an equivalent definition in terms of such choices of $U$:
\begin{equation}\label{eq: defF}
S_{\epsilon, \delta}(X) \defn \sup_{U ~:~ p_{X|U}(\cdot|U) \in \cP(\mathbf r)}  \sum_{i=1}^k d_i h(X_i + \sqrt \delta W_i | U) - \sum_{j =1}^m c_j h(A_j (X +\sqrt \delta W) + \sqrt \epsilon Z_j \mid U),
\end{equation} 
where, on the right hand side of equation \eqref{eq: defF}, we can assume that $W$, $Z$ and $(U, X)$ are mutually independent. For a particular choice of $U$, define 
\begin{equation}\label{eq: s_given_u}
s_{\epsilon, \delta}(X \mid U) := \sum_{i=1}^k d_i h(X_i + \sqrt \delta W_i| U) - \sum_{j=1}^m c_j h(A_j (X+\sqrt \delta W) + \sqrt \epsilon Z_j | U).
\end{equation}

Analogous to $\cP(\mathbf r)$, define $\cP(2\mathbf r)$ to be the set of random variables that take values in $\real^{2r_1} \times \dots \times \real^{2r_k}$ and satisfy the conditions in Definition \ref{def: pg}. 
More precisely, a random vector $(X_1, X_2)$ is in $\cP(2  \mathbf r)$
if $X_1 := (X_{11}, \ldots, X_{k1})$
and $X_2 := (X_{12}, \ldots, X_{k2})$
are $\real^{n}$-valued random vectors such that the random vectors $(X_{i1}, X_{i2}) \in \real^{2r_i}$, $i \in [k]$ are mutually independent, satisfy the condition in equation \eqref{eq: CarEtal_condition}, and condition 3 of Definition \ref{def: pg} holds for $(X_1, X_2)$. Since the condition in equation \eqref{eq: CarEtal_condition}
is inherited by marginalization, we have that if $(X_1, X_2) \in \cP(2  \mathbf r)$ then $X_1 \in \cP(\mathbf r)$
and $X_2 \in \cP(\mathbf r)$.

We will need to define an extension of  $S_{\epsilon, \delta}$ to the larger space $\cP(2\mathbf r)$. Consider a random vector $(X_1, X_2) \in \cP(2  \mathbf r)$
as in the preceding paragraph. Define 
\begin{align}\label{eq: s_y1y2}
s_{\epsilon, \delta}(X_1, X_2)   &\defn \sum_{i=1}^k d_i h(X_{i1}+\sqrt \delta W_{i1}, X_{i2} + \sqrt \delta W_{i2} ) \nonumber \\
& \qquad - \sum_{j=1}^m c_j h(A_j (X_{1} + \sqrt \delta W_1) + \sqrt \epsilon Z_{j1}, A_j (X_{2}+ \sqrt \delta W_2) + \sqrt \epsilon Z_{j2}),
\end{align}
where $(W_1, W_2, Z_1, Z_2)$ are mutually independent standard normal distributions of the appropriate dimensions that are independent of $(X_1, X_2)$.
The concave envelope of $s_{\epsilon, \delta}(X_1,X_2)$ 
can be written as:
\begin{align}\label{eq: y11}
S_{\epsilon, \delta}(X_1, X_2) &= \sup_{U ~:~ p_{X_1 X_2|U}(\cdot, \cdot|U) \in \cP(2\mathbf r)} \sum_{i=1}^k d_i h(X_{i1}+\sqrt \delta W_{i1}, X_{i2} + \sqrt \delta W_{i2} \mid U ) \nonumber \\
&- \sum_{j=1}^m c_j h(A_j (X_{1} + \sqrt \delta W_1) + \sqrt \epsilon Z_{j1}, A_j (X_{2}+ \sqrt \delta W_2) + \sqrt \epsilon Z_{j2} \mid U),
\end{align}
where $W_1$, $W_2$, $Z_1$, $Z_2$ and $(U, X_1, X_2)$ are mutually independent, with $U$ taking values in  finite sets $\cU$ and 
$p_{X_1,X_2|U}(\cdot, \cdot|U) \in  \cP(2\mathbf r)$.
Figure \ref{fig:gmodel} illustrates the relations between the random variables via a graphical model. 

\begin{figure}[htb] \centering
\begin{tikzpicture}
\node at(1.5,3) {$X_1$};
\node at(2.5,3) {$X_2$};
\node at(1.5,2.2) {$X_{11}$};
\node at(2.5,2.2) {$X_{12}$};
\node at(1.5,1.2) {$X_{21}$};
\node at(2.5,1.2) {$X_{22}$};
\node at(1.5,-0.5) {$\vdots$};
\node at(2.5,-0.5) {$\vdots$};
\node at(1.5,-2.2) {$X_{k1}$};
\node at(2.5,-2.2) {$X_{k2}$};
\node at(-2,0) {$U$};
\draw[rounded corners=5pt,color=blue]
  (1.1,-2.5) rectangle ++(0.8,5);
  \draw[rounded corners=5pt,color=blue]
  (2.1,-2.5) rectangle ++(0.8,5);
  \draw[thick,->] (1.8, -2.2) -- (2.2, -2.2);
  \draw[thick,->] (1.8, 1.2) -- (2.2, 1.2);
  \draw[thick,->] (1.8, 2.2) -- (2.2, 2.2);
  \draw[thick,->] (-1.8, 0.3) -- (1.2, 2.2);
  \draw[thick,->] (-1.8, 0.1) -- (1.2, 1.2);
  \draw[thick,->] (-1.8, -0.3) -- (1.2, -2.2);
  \draw[thick,->] (-1.8, 0.26) -- (2.4, 2.0);
  \draw[thick,->] (-1.8, 0.06) -- (2.4, 1.0);
  \draw[thick,->] (-1.8, -0.26) -- (2.4, -2.0);
\end{tikzpicture}
\caption{Illustration of the Markov relationship}
\label{fig:gmodel}
\end{figure}
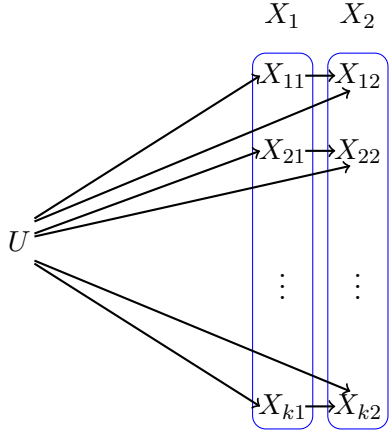

\subsubsection{Proof of subadditivity}\label{section: tensorization: proof}
\begin{lemma}[Subadditivity lemma]\label{lemma: EPI+BL+sub}
For any $\epsilon, \delta \geq 0$, the function $S_{\epsilon, \delta}$ is subadditive; i.e., if $(X_1, X_2) \in \cP(2 \mathbf r)$ then 
\begin{equation}
S_{\epsilon, \delta}(X_1, X_2) \leq S_{\epsilon, \delta}(X_1) + S_{\epsilon, \delta}(X_2).
\end{equation}
\end{lemma}
\begin{corollary}\label{cor: EPI+BL}
For any $\epsilon, \delta \geq 0$, the function $S_{\epsilon, \delta}$ tensorizes; i.e., if $X_1, X_2 \in \cP(\mathbf r)$ and if $X_1 \ind X_2$, then
\begin{equation}
S_{\epsilon, \delta}(X_1, X_2) = S_{\epsilon, \delta}(X_1) + S_{\epsilon, \delta}(X_2).
\end{equation}
\end{corollary}

\begin{proof}[Proof of Lemma \ref{lemma: EPI+BL+sub}]

Let $U$ be an auxiliary random variable
taking values in a finite set $\cU$, such that  $p_{X_1,X_2 \mid U} (\cdot, \cdot|U) \in \cP(2 \mathbf r)$. Consider the following expansion, which comes from applying expansion (A) term by term: 
\begin{align}
s_{\epsilon, \delta}(X_1, X_2 \mid U)
&= \Bigg[ \sum_{i=1}^k d_i h(X_{i1} + \sqrt \delta W_{i1} | U ) - \sum_{j=1}^m c_j h(A_j (X_{1}+\sqrt \delta W_1) + \sqrt \epsilon Z_{j1} | U) \Bigg] \nonumber\\
& \quad + \Bigg[\sum_{i=1}^k d_i h( X_{i2} + \sqrt \delta W_{i2}| U ) - \sum_{j=1}^m c_j h(A_j (X_{2} + \sqrt \delta W_2) + \sqrt \epsilon Z_{j2}| U) \Bigg] \nonumber\\
& \quad + \Bigg[ - \sum_{i=1}^k d_i I(X_{i1} + \sqrt \delta W_{i1}; X_{i2}+ \sqrt \delta W_{i2} | U) \nonumber \\
&\qquad \qquad + \sum_{j=1}^m c_j I(A_j (X_{1}+\sqrt \delta W_1) + \sqrt \epsilon Z_{j1}; A_j (X_{2} + \sqrt \delta W_2) + \sqrt \epsilon Z_{j2} | U)\Bigg]. \label{eq: exp1}
\end{align}

For simplicity, denote the terms in the square brackets by $T_1(U)$, $T_2(U)$, and $T_3(U)$, respectively,
even though they actually depend on $p_{U|X_1,X_2}$.
Observe that   that $p_{X_1|U}(\cdot|U), p_{X_2|U}(\cdot|U) \in \cP(\mathbf r)$ (see Figure \ref{fig:gmodel}). Thus, we conclude that $T_1(U) \leq S_{\epsilon, \delta}(X_1)$ and $T_2(U) \leq S_{\epsilon, \delta}(X_2)$, using the definition in equation \eqref{eq: defF}. Substituting these inequalities, we arrive at 
\begin{align}
s_{\epsilon, \delta}(X_1, X_2 | U) \leq S_{\epsilon, \delta}(X_1) + S_{\epsilon, \delta}(X_2) + T_3(U).\label{eq: f2}
\end{align}
We now expand 
$s_{\epsilon, \delta}(X_1,X_2 \mid U) $
in a different way, which comes from applying expansion (B) term by term:
{\small \begin{align}
&s_{\epsilon, \delta}(X_1, X_2 \mid U)\\ 
&= \Big[ \sum_{i=1}^k d_i h(X_{i1} + \sqrt \delta W_{i1} | U, X_{i2} + \sqrt \delta W_{i2} ) - \sum_{j=1}^m c_j h(A_j (X_{1}+\sqrt \delta W_1)+\sqrt \epsilon Z_{j1} | U, A_j (X_{2}+\sqrt \delta W_2)+\sqrt \epsilon Z_{j2}) \Big] \nonumber\\
& + \Big[\sum_{i=1}^k d_i h( X_{i2} + \sqrt \delta W_{i2} | U, X_{i1} + \sqrt \delta W_{i1} ) - \sum_{j=1}^m c_j h(A_j (X_{2}+\sqrt \delta W_2)+\sqrt \epsilon Z_{j2} | U, A_j (X_{1}+\sqrt \delta W_1)+\sqrt \epsilon Z_{j1}) \Big] \nonumber\\
& + \Big[ \sum_{i=1}^k d_i I(X_{i1} + \sqrt \delta W_{i1}; X_{i2} + \sqrt \delta W_{i2} | U) - \sum_{j=1}^m c_j I(A_j (X_{1}+\sqrt \delta W_1)+\sqrt \epsilon Z_{j1}; A_j (X_{2}+\sqrt \delta W_2)+\sqrt \epsilon Z_{j2} | U)\Big]. \label{eq: exp2}
\end{align}}
For ease of notation, call the three terms in the square brackets $R_1(U)$, $R_2(U)$, and $R_3(U) = -T_3(U)$, respectively,
even though each term actually depends on $p_{U|X_1,X_2}$.
Similar to inequality \eqref{eq: f2}, we would like to upper bound $R_1(U)$ and $R_2(U)$ by $S_{\epsilon, \delta}(X_1)$ and $S_{\epsilon, \delta}(X_2)$ respectively. However, the conditioning in each of the two differential entropy terms in each 
$R_a(U)$, $a = 1,2$ is not the same, so we cannot directly conclude such a bound. Using the chain rule of mutual information and data-processing relations, we may make the conditioning in $R_1(U)$ and $R_2(U)$ uniform by introducing some extra mutual information terms:
\begin{align*}
&R_1(U)\\ 
&= \Big[ \sum_{i=1}^k d_i h(X_{i1} + \sqrt \delta W_{i1} | U, X_{i2} + \sqrt \delta W_{i2} ) - \sum_{j=1}^m c_j h(A_j (X_{1}+\sqrt \delta W_1)+\sqrt \epsilon Z_{j1} | U, A_j (X_{2}+\sqrt \delta W_2)+\sqrt \epsilon Z_{j2}) \Big] \\
&= \Big[ \sum_{i=1}^k d_i h(X_{i1} + \sqrt \delta W_{i1} | U, X_{2} + \sqrt \delta W_{2} ) - \sum_{j=1}^m c_j h(A_j (X_{1}+\sqrt \delta W_1)+\sqrt \epsilon Z_{j1} | U, X_{2}+\sqrt \delta W_2) \Big]\\
&- \sum_{j=1}^m c_jI(  A_j (X_1+\sqrt \delta W_1) + \sqrt \epsilon Z_{j1} ; X_2+\sqrt \delta W_2 | U, A_j (X_2 + \sqrt \delta W_2)+\sqrt \epsilon Z_{j2})\\
&=: \tilde R_1(U) - I_1(U),
\end{align*}
where we write $\tilde R_1(U)$
and $I_1(U)$ for simplicity, even thought the corresponding terms depend on $p_{U|X_1,X_2}$.
The above steps are justified as follows. First, it is easy to check that $(X_{i1} + \sqrt \delta W_{i1}) \ind \{X_{l2} + \sqrt \delta W_{l2}\}_{l \neq i}$ conditioned on $(U, X_{i2} + \sqrt \delta W_{i2})$. This means that, for all $1 \leq i \leq k$, 
\begin{align*}
&h(X_{i1} + \sqrt \delta W_{i1} | U, X_{i2} + \sqrt \delta W_{i2} )\\
&= h(X_{i1} + \sqrt \delta W_{i1} | U, X_{12} + \sqrt \delta W_{12}, \dots, X_{i2} + \sqrt \delta W_{i2}, \dots, X_{k2} + \sqrt \delta W_{k2} )\\
&= h(X_{i1} + \sqrt \delta W_{i1} | U, X_2 + \sqrt \delta W_2 ).
\end{align*}
Also, we may verify the Markov chain (conditioned on $U$)
\begin{align*}
\Big[A_j (X_2+\sqrt \delta W_2) + \sqrt \epsilon Z_{j2}\Big] \rt \Big[X_2+\sqrt \delta W_2\Big] \rt  \Big[A_j (X_1+ \sqrt \delta W_1) + \sqrt \epsilon Z_{j1}\Big],
\end{align*}
which gives the equality
\begin{align*}
&h(A_j (X_{1}+\sqrt \delta W_1) + \sqrt \epsilon Z_{j1} | U, A_j (X_2+\sqrt \delta W_2) + \sqrt \epsilon Z_{j2})\\
 &= h(A_j (X_{1} + \sqrt \delta W_1) + \sqrt \epsilon Z_{j1} | U, X_2+\sqrt \delta W_2) \\
 &+ I(  A_j (X_1+ \sqrt \delta W_1) + \sqrt \epsilon Z_{j1}; X_2+\sqrt \delta W_2 | U, A_j (X_2+\sqrt \delta W_2) + \sqrt \epsilon Z_{j2}).
\end{align*}
Similar reasoning for $R_2(U)$ gives 
\begin{align*}
&R_2(U)\\ 
&= \Big[ \sum_{i=1}^k d_i h(X_{i2} + \sqrt \delta W_{i2} | U, X_{i1} + \sqrt \delta W_{i1} ) - \sum_{j=1}^m c_j h(A_j (X_{2}+\sqrt \delta W_2)+\sqrt \epsilon Z_{j2} | U, A_j (X_{1}+\sqrt \delta W_1)+\sqrt \epsilon Z_{j1}) \Big] \\
&= \Big[ \sum_{i=1}^k d_i h(X_{i2} + \sqrt \delta W_{i2} | U, X_{1} + \sqrt \delta W_{1} ) - \sum_{j=1}^m c_j h(A_j (X_{2}+\sqrt \delta W_2)+\sqrt \epsilon Z_{j2} | U, X_{1}+\sqrt \delta W_1) \Big]\\
&- \sum_{j=1}^m c_jI(  A_j (X_2+\sqrt \delta W_2) + \sqrt \epsilon Z_{j2} ; X_1+\sqrt \delta W_1 | U, A_j (X_1 + \sqrt \delta W_1)+\sqrt \epsilon Z_{j1})\\
&=: \tilde R_2(U) - I_2(U),
\end{align*}
where we use the notation $\tilde R_2(U)$ and $I_2(U)$ for simplicity, even though the corresponding terms depend on $p_{U|X_1,X_2}$.
Substituting the expressions for $R_1(U)$ and $R_2(U)$ in the expansion in equation \eqref{eq: exp2}, we arrive at 
\begin{align}
s_{\epsilon, \delta}(X_1, X_2 \mid U) &= \tilde R_1(U) + \tilde R_2(U) - T_3(U) - I_1(U) - I_2(U) \nonumber \\
&\stackrel{(a)}\leq S_{\epsilon, \delta}(X_1) + S_{\epsilon, \delta}(X_2) - T_3(U) - I_1(U) - I_2(U) \nonumber\\
&\stackrel{(b)}\leq S_{\epsilon, \delta}(X_1) + S_{\epsilon, \delta}(X_2) - T_3(U). \label{eq: f3}
\end{align}
Here, in step $(a)$ we used the fact that $p_{X_1|U,X_2+\sqrt \delta W_2}(\cdot|U,X_2+\sqrt \delta W_2), p_{X_2|U,X_{1}+\sqrt \delta W_1}(\cdot|U,X_{1}+\sqrt \delta W_1) \in \cP(\mathbf r)$ 
and the definition in equation \eqref{eq: defF}. Step $(b)$ follows by noticing that the $c_j$ are non-negative, and so are $I_1(U)$ and $I_2(U)$ since they are nonnegative linear combinations of mutual informations.

We can combine inequalities
\eqref{eq: f2}
and \eqref{eq: f3} to get
\begin{equation}
s_{\epsilon, \delta}(X_1, X_2|U) \leq S_{\epsilon, \delta}(X_1) + S_{\epsilon, \delta}(X_2).
\end{equation}
Taking the supremum on the left hand side of this inequality over all auxiliary variables 
$U$
taking values in finite sets $\cU$, such that  $p_{X_1,X_2 \mid U} (\cdot, \cdot|U) \in \cP(2 \mathbf r)$,
yields the claimed subadditivity result.
\end{proof}

\begin{proof}[Proof of Corollary \ref{cor: EPI+BL}]
When $X_1 \ind X_2$, we have the inequality 
\begin{equation}\label{eq: f1}
S_{\epsilon, \delta}(X_1, X_2) \geq S_{\epsilon, \delta}(X_1) + S_{\epsilon, \delta}(X_2).
\end{equation}
This is because we can always choose $U \defn (U_1, U_2)$ such that $(U_1, X_1) \ind (U_2, X_2)$ and $p_{X_1|U_1}(\cdot|U_1)$, $p_{X_2|U_2}(\cdot|U_2) \in \cP(\mathbf r)$. The supremum in equation \eqref{eq: y11} over this restricted class of auxiliaries is simply $S_{\epsilon, \delta}(X_1) + S_{\epsilon, \delta}(X_2)$, which therefore is a  lower bound on $S_{\epsilon, \delta}(X_1, X_2)$. 
Inequality \eqref{eq: f1} combined with Lemma \ref{lemma: EPI+BL+sub} completes the proof of Corollary \ref{cor: EPI+BL}.
\end{proof}

Our next lemma serves to some extent as a converse to Corollary \ref{cor: EPI+BL}. In particular, we show that if $S_{\epsilon, \delta}(X_1, X_2) = S_{\epsilon, \delta}(X_1) + S_{\epsilon, \delta}(X_2)$, then $X_1$ and $X_2$ are independent conditioned on the optimal auxiliary $U^*$, assuming it exists. We point out that this converse requires $\epsilon$ and $\delta$ to be strictly bounded away from $0$, unlike Lemma~\ref{lemma: EPI+BL+sub}. The formal statement is as follows:
 
\begin{lemma}[Independence relations]\label{lemma: independence}
Fix $\epsilon, \delta > 0$.
Given $(X_1, X_2) \in \cP(2\mathbf r)$, suppose that $S_{\epsilon, \delta}(X_1, X_2) =  S_{\epsilon, \delta}(X_1) + S_{\epsilon, \delta}(X_2)$. Suppose that $U^*$ is such that $p_{X_1, X_2|U^*}(\cdot, \cdot|U^*) \in \cP(2\mathbf r)$ and $s_{\epsilon, \delta}(X_1, X_2 \mid U^*) =  S_{\epsilon, \delta}(X_1, X_2)$.  Then the following results hold:
\begin{enumerate}
    \item [(a)] For all $u^* \in \cU^*$, we have that $X_1 \ind X_2$ conditioned on $U^* = u^*$,
    \item[(b)] $s_{\epsilon, \delta}(X_1 | U^*) = S_{\epsilon, \delta}(X_1)$ and $s_{\epsilon, \delta}(X_2 | U^*) = S_{\epsilon, \delta}(X_2)$.
\end{enumerate} 

\end{lemma}
\begin{proof}
Notice that the proof of Lemma \ref{lemma: EPI+BL+sub} implies that the optimizing $U^*$, if it exists, must satisfy $I_1(U^*) = I_2(U^*) = T_3(U^*) = 0$.
The first two equalities yield  the Markov chains (conditioned on $U^* = u^*$)
\begin{align*}
&\Big[A_j (X_1+ \sqrt \delta W_1) + \sqrt \epsilon Z_{j1}\Big] \rt \Big[A_j (X_2+\sqrt \delta W_2) + \sqrt \epsilon Z_{j2}\Big] \rt \Big[X_2+\sqrt \delta W_2\Big], \quad \text{ and }\\
&\Big[A_j (X_2+ \sqrt \delta W_2) + \sqrt \epsilon Z_{j2}\Big] \rt \Big[A_j (X_1+\sqrt \delta W_1) + \sqrt \epsilon Z_{j1}\Big] \rt \Big[X_1+\sqrt \delta W_1\Big].
\end{align*}
However, we have the obvious Markov chains
\begin{align*}
&\Big[A_j (X_1+ \sqrt \delta W_1) + \sqrt \epsilon Z_{j1}\Big] \rt  \Big[X_2+\sqrt \delta W_2\Big] \rt \Big[A_j (X_2+\sqrt \delta W_2) + \sqrt \epsilon Z_{j2}\Big], \quad \text{ and }\\
&\Big[A_j (X_2+ \sqrt \delta W_2) + \sqrt \epsilon Z_{j2}\Big] \rt  \Big[X_1+\sqrt \delta W_1\Big] \rt \Big[A_j (X_1+\sqrt \delta W_1) + \sqrt \epsilon Z_{j1}\Big].
\end{align*}
Using Lemma \ref{lemma: markov_ind}, we may conclude that, conditioned on $U^*$, we have
\begin{align*}
\Big[A_j (X_1+ \sqrt \delta W_1) + \sqrt \epsilon Z_{j1}\Big] &\ind \Big[X_2+\sqrt \delta W_2\Big], \quad \text{ and }\\
\Big[A_j (X_2+ \sqrt \delta W_2) + \sqrt \epsilon Z_{j2}\Big] &\ind \Big[X_1+\sqrt \delta W_1\Big].
\end{align*}
Recall that $T_3(U^*)$ is given by
\begin{align*}
& T_3(U^*)= \Bigg[ - \sum_{i=1}^k d_i I(X_{i1} + \sqrt \delta W_{i1}; X_{i2}+ \sqrt \delta W_{i2} | U^*) \\
& \qquad \qquad \qquad + \sum_{j=1}^m c_j I(A_j (X_{1}+\sqrt \delta W_1) + \sqrt \epsilon Z_{j1}; A_j (X_{2} + \sqrt \delta W_2) + \sqrt \epsilon Z_{j2} | U^*)\Bigg].
\end{align*}
Substituting the above independence relations in $T_3(U^*) = 0$, we conclude that, conditioned on $U^*$, we have
\begin{align*}
X_1 + \sqrt \delta W_1 \ind X_2 + \sqrt \delta W_2,
\end{align*}
which by Lemma \ref{lemma: xz_ind} implies that, conditioned on $U^*$, we have
\begin{align*}
X_1  \ind X_2,
\end{align*}
and concludes the proof of (a). 

Having proved (a), rewrite equation \eqref{eq: exp1}, with $U^*$ for $U$, as
\begin{equation}
    s_{\epsilon, \delta}(X_1, X_2|U^*) = s_{\epsilon, \delta}(X_1|U^*) + s_{\epsilon, \delta}(X_2|U^*).
\end{equation}
The above inequality, combined with the assumed equality $s_{\epsilon, \delta}(X_1, X_2|U^*) = S_{\epsilon, \delta}(X_1) + S_{\epsilon, \delta}(X_2)$, immediately yields 
\begin{align*}
    s_{\epsilon, \delta}(X_1|U^*) &= S_{\epsilon, \delta}(X_1) \quad, \text{ and }\\
    s_{\epsilon, \delta}(X_2|U^*) &= S_{\epsilon, \delta}(X_2).
\end{align*}
\end{proof}

\subsubsection{A general subadditivity result}\label{section: tensorization: general}
A closer inspection of the proof of Lemma~\ref{lemma: EPI+BL+sub} reveals that the linear functions mapping $X$ to $A_j X$ could be replaced with general channels. To be precise, let $X = (X_1, X_2, \dots, X_k) \in \cP(\mathbf r)$ and for $j \in [m]$, consider $m$ channels $p_{Y_j | X}$ from $X$ to $Y_j$. Define the function $s: \cP(\mathbf r) \to \real$ as
\begin{align}
s(X) := \sum_{i=1}^k d_i h(X_i) - \sum_{j=1}^m c_jh(Y_j),
\end{align}
and let $S$ be its concave envelope. The function $s$ is lifted to pairs of random variables $(X_1, X_2) \in \cP(2\mathbf r)$ as
\begin{align*}
s(X_1, X_2) := \sum_{i=1}^k d_i h(X_{i1}, X_{i2}) - \sum_{j=1}^m c_j h(Y_{j1}, Y_{j2}),
\end{align*}
where the channel from $(X_1, X_2)$ to $(Y_{j1}, Y_{j2})$ is given by $p_{Y_{j1}, Y_{j2} | X_1, X_2} = p_{Y_{j1}| X_1}p_{Y_{j2} | X_2}$. Let $S(X_1, X_2)$ be the concave envelope of $s(X_1, X_2)$.
\begin{claim}\label{claim: sub-general}
The function $S$ is subadditive; i.e., $S(X_1, X_2) \leq S(X_1) + S(X_2)$.
\end{claim}
\begin{proof}
Let $U$ be an auxiliary random variable taking values in a finite set $\cU$, such that  $p_{X_1,X_2 \mid U} (\cdot, \cdot|U) \in \cP(2 \mathbf r)$. Note that
\begin{align*}
&\sum_{i=1}^k d_i h(X_{i1}, X_{i2} | U) - \sum_{j=1}^m c_j h(Y_{j1}, Y_{j2} | U)\\
&\quad = \sum_{i=1}^k d_i h(X_{i1} | U) - \sum_{j=1}^m c_j h(Y_{j1} | U) + \sum_{i=1}^k d_i h(X_{i2} | U, X_{i1}) - \sum_{j=1}^m c_j h(Y_{j2} | U, Y_{j1})\\
&\quad \stackrel{(a)}\leq \sum_{i=1}^k d_i h(X_{i1} | U) - \sum_{j=1}^m c_j h(Y_{j1} | U) + \sum_{i=1}^k d_i h(X_{i2} | U, X_{1}) - \sum_{j=1}^m c_j h(Y_{j2} | U, X_{1}) \\
&\quad \leq S(X_1) + S(X_2).
\end{align*}

To verify step (a) it suffices to show that $h(X_{i2} | U, X_{i1}) \le h(X_{i2} | U, X_{1})$
for each $i \in [k]$. In fact we have equality here because, as is easily verified,
we have $(X_{l1}, l \neq i)$ conditionally independent of
$X_{i2}$ given $X_{i1}$ and $U$.
To verify the last inequality, observe that $(X_{i1}, 1 \le i \le k)$ are conditionally independent given $U$
and $(X_{i2}, 1 \le i \le k)$ are conditionally independent given $(U,X_1)$.
Taking a supremum over $U$ completes the proof.
\end{proof}

We make several remarks. First, observe that only the $c_j$ need to be non-negative; no such condition is necessary for the $d_i$.
\footnote{
 However, studying the maximum over $\cP(\mathbf r)$ of an expression like \eqref{eq: diff_h} when some of the $d_i$ are negative is not interesting because the maximum over $\cP_g(\mathbf r)$ is $\infty$, as can be seen by letting the covariance matrix of the component corresponding to any factor with negative $d_i$ tend to $0$.
}
Second, while this proof is very simple compared to that of Lemma~\ref{lemma: EPI+BL+sub},
the independence relations in Lemma~\ref{lemma: independence}---which are critical to the proof of Gaussian optimality---cannot be directly deduced from the above proof. However, this is not such a big impediment. Instead of $s(X)$, we could consider a slightly modified function $s_\epsilon(X)$ defined by
\begin{align*}
s_\epsilon(X) = \sum_{i=1}^k d_i h(X_i) - \sum_{j=1}^m c_jh(Y_j) - \epsilon I(X; X+Z),
\end{align*}
where $Z$ is a standard Gaussian that is independent of $X$ and $Y_j = A_jX$. It is not hard to show that the concave envelope of $s_\epsilon$ is subadditive; in fact, the same steps as in the proof of Claim~\ref{claim: sub-general} suffice. Further,
including the extra mutual information term allows one to deduce independence relations analogous to those in Lemma~\ref{lemma: independence}. This approach provides an alternate route to proving Theorem~\ref{thm: EPI+BL}.

\subsection{Proof of Theorem \ref{thm: EPI+BL}}\label{section: proof}

Having proved the key subadditivity step, the rest of the proof closely follows the steps outlined in \cite[Appendix II]{GengNair14}. 

\begin{definition}\label{def: vV}
Let $\Sigma \defn \text{Diag}(\Sigma_1, \Sigma_2, \dots, \Sigma_k)$ be an $n \times n$ block diagonal matrix such that each $\Sigma_i$ is an $r_i \times r_i$ positive definite matrix. For $\epsilon, \delta > 0$, define
\begin{align}
v (\Sigma) &:= \sup_{X \in \cP(\mathbf r), \E XX^T = \Sigma} s_{\epsilon, \delta}(X), \quad \text{ and } \label{eq: vV-v}\\
V(\Sigma) &:= \sup_{X \in \cP(\mathbf r), \E XX^T \preceq \Sigma} S_{\epsilon, \delta}(X), \label{eq: vV-V}
\end{align}
where $\preceq$ denotes ordering in the positive semidefinite partial order.
\end{definition}

\begin{lemma}\label{lemma: xstar_ustar}
There exist random variables $X^*$ and $U^*$ satisfying (1) $| \cU^*| \leq \sum_{i=1}^k \frac{r_i(r_i+1)}{2} +1$; (2) $X^* \in \cP(\mathbf r)$; and (3) $\E X^* {X^*}^T \preceq \Sigma$, such that the following holds:
\begin{equation}
V(\Sigma) = s_{\epsilon, \delta}(X^* \mid U^*).
\end{equation}
\end{lemma}
\begin{proof}[Proof of Lemma \ref{lemma: xstar_ustar}]

Let $(X^{(t)}, t \ge 1)$ be a sequence of random variables such that $\E X^{(t)} (X^{(t)})^T = \widehat \Sigma$ and $s_{\epsilon, \delta}(X^{(t)}) \uparrow v (\widehat \Sigma)$ as $t \to \infty$. This sequence of random variables is tight due to the covariance constraint \cite[Proposition 17]{GengNair14}, and thus we may assume without loss of generality that the $X^{(t)}$ converge weakly to a random variable $X^{\widehat \Sigma}$ as $t \to \infty$. Since $X^{(t)} + \sqrt \delta W$ satisfies the necessary regularity conditions as in \cite[Proposition 18]{GengNair14}, we also have $h(X^{(t)}_{i} + \sqrt \delta W_i) \to h(X^{\widehat \Sigma}_i + \sqrt \delta W_i)$ for $i \in [k]$, and $h(A_j (X^{(t)} +\sqrt \delta W) + \sqrt \epsilon Z_j ) \to h(A_j (X^{\widehat \Sigma} +\sqrt \delta W) + \sqrt \epsilon Z_j )$ for $j \in [m]$. Hence we may conclude  $s_{\epsilon, \delta}(X^{\widehat \Sigma}) = v(\widehat \Sigma)$.

Recall that $V(\Sigma)$ is defined as
\begin{align}
V(\Sigma) &= \sup_{X \in \cP(\mathbf r), \E XX^T \preceq \Sigma} S_{\epsilon, \delta}(X) \nonumber \\
&= \sup_{(U, X), p_{X|U}(\cdot|U) \in \cP(\mathbf r), \E XX^T \preceq \Sigma} s_{\epsilon, \delta}(X \mid U) \nonumber \\
&\stackrel{(a)}= \sup_{\alpha_l \geq 0, \widehat \Sigma_l : \sum_{l=1}^M \alpha_l = 1, \sum_{l=1}^M \alpha_l \widehat \Sigma_l \preceq \Sigma} \sum_{l=1}^M \alpha_l v(\widehat \Sigma_l), \label{eq: Caratheodory}
\end{align}
where, for the moment, $M$
ranges over positive integers of
arbitrary size.
The equality in $(a)$ is because we may restrict $p_{X|U}(\cdot|U)$ to the class of optimizers $X^{\widehat \Sigma}$ for $\widehat \Sigma \succeq 0$. 
We now show that we can fix $M$ to be $\sum_{i=1}^k {{r_i+1} \choose 2}+1$
in \eqref{eq: Caratheodory}.
Let $\mathcal{T}$ denote the connected subset of positive definite matrices $\Sigma$ of the form $\text{Diag}(\Sigma_{1}, \dots,  \Sigma_{k})$  where $\Sigma_{i}$ is an $r_i \times r_i$ positive definite matrix for $i \in [k]$. Consider the
connected compact subset, $\mathcal{V}$, of the $M$-dimensional Euclidean space obtained using the continuous mapping $\Phi: \mathcal{T} \mapsto \mathbb{R}^{M},  $ defined by $\Phi(\Sigma) = \left(\{\Sigma_i (j,k)_{1\leq j \leq k \leq r_i}\}, v(\Sigma)\right)$, where $M := \sum_{i=1}^k {{r_i+1} \choose 2}+1$. Fenchel's extension of Carath\'{e}odory's Theorem \cite[Theorem 1.3.7]{HUT} states that any finite convex combination of points in $\mathcal{V}$, can be represented as a convex combination of at most $M$ points in $\mathcal{V}$. Hence for any $(U,X^{\Sigma_U})$ we can find a pair $(U',X^{\Sigma_{U'}})$ with 
$U'$ taking at most $M$ values,
such that $E(\Sigma_U) = E(\Sigma_{U'})$ and $E(v(\Sigma_U)) = E(v(\Sigma_{U'}))$. 
Thus from this point onwards in the proof we define
$M := \sum_{i=1}^k {{r_i+1} \choose 2}+1$ in \eqref{eq: Caratheodory}.

Consider any sequence of convex combinations $\left( \{\alpha_l^{(t)}\}_{l=1}^M, \{\widehat \Sigma_l^{(t)}\}_{l=1}^M \right)$ 
with $\sum_{l=1}^M \alpha_l^{(t)} \widehat \Sigma_l^{(t)} \preceq \Sigma$ for all $t \ge 1$, and such that $\sum_{l=1}^M \alpha_l^{(t)} v(\widehat \Sigma_l^{(t)})$ converges to $v(\Sigma)$ as
$t \to \infty$. Appealing to the compactness of the $M$-dimensional simplex, we may assume without loss of generality that $\alpha_l^{(t)} \to \alpha^*_l$ for all $i \in [M]$. If any of the $\alpha^*_l$ equals $0$,  then noticing that $\alpha_l^{(t)} \widehat \Sigma_l^{(t)} \preceq \Sigma$ gives us
\begin{align*}
v(\widehat \Sigma_l^{(t)}) &\stackrel{(a)}\leq \sum_{i=1}^k \frac{d_i}{2} \log (2\pi e)^{r_i}|\widehat \Sigma_{li}^{(t)} +  \delta I_{r_i \times r_i}| - \sum_{j=1}^m \frac{c_j n_j}{2} \log (2 \pi e \epsilon) \\
&\leq  \sum_{i=1}^k \frac{d_i}{2} \log (2\pi e)^{r_i} \Bigg| \frac{\Sigma_{li}}{\alpha^{(t)}_l}+  \delta I_{r_i \times r_i}\Bigg| - \sum_{j=1}^m \frac{c_j n_j}{2} \log (2 \pi e \epsilon)\\
&=  \sum_{i=1}^k \frac{d_i}{2} \log \Bigg| \frac{\Sigma_{li}}{\alpha^{(t)}_l} +  \delta I_{r_i \times r_i}\Bigg|+ C_0, 
\end{align*}
where $C_0$ is some constant that does not depend on $t$. 
In $(a)$, we used the fact that each $h(X_i+ \sqrt \delta W_i)$ is upper-bounded by the entropy of a Gaussian random variable with the same covariance matrix as $X_i +\sqrt \delta W_i$, and $h(A_j (X+\sqrt \delta W) + \sqrt \epsilon Z_j) \geq h(\sqrt \epsilon Z_j)$.

 It is now clear that the limit $\alpha_l^{(t)} v(\widehat \Sigma^{(t)}_l)$ as $t \to \infty$ is equal to 0 whenever $\alpha_l^{(t)} \to 0$. Thus, we may assume that $\min_{l \in [M]} \alpha^*_l = \alpha_{\min} > 0$, by splitting a component $\alpha_l^{(t)} v(\widehat \Sigma_l^{(t)})$ into multiple components if necessary.
This implies that $\widehat \Sigma^{(t)}_l \preceq \frac{2\Sigma}{\alpha_{\min}}$ for all large enough $t$. Hence, we can find a convergent subsequence such that $\widehat \Sigma^{(t)}_l \rt \Sigma_l^*$ for each $l \in [M]$ when $t \to \infty$ along this subsequence. We arrive at
\begin{equation}
V(\Sigma) = \sum_{l=1}^M \alpha_l^* v(\Sigma^*_l),
\end{equation}
or, in other words, we can find a pair of random variables $(X^*, U^*)$ with $|\cU^*| \leq M$ such that $V(\Sigma) = s_{\epsilon, \delta}(X^* | U^*)$. This completes the proof.
\end{proof}

\begin{lemma}\label{lemma: rotation}
Consider random variables $(X_1,X_2, U)$ such that $(X_1, X_2) \in \cP(2\mathbf r)$ for some $\mathbf r$-partition of $n > 0$. Define new random variables $X_+$ and $X_-$ via
\begin{align*}
X_+ := \frac{X_1 + X_2}{\sqrt 2}, \quad &\text{ and } \quad X_- := \frac{X_1 - X_2}{\sqrt 2}.
\end{align*}
Then $s_{\epsilon, \delta}(X_1, X_2 | U) = s_{\epsilon, \delta}(X_+, X_-|U)$.
\end{lemma}
\begin{proof}
We have the equality
\begin{align}
s_{\epsilon, \delta}(X_1, X_2|U) &=  \sum_{i=1}^k d_i h(X_{i1}+\sqrt \delta W_{i1}, X_{i2}+\sqrt \delta W_{i2} | U ) \nonumber \\
&- \sum_{j=1}^m c_j h(A_j(X_1+\sqrt \delta W_1) + \sqrt \epsilon Z_{j1}, A_j(X_2+\sqrt \delta W_2) + \sqrt \epsilon Z_{j2} | U),
\end{align}
Further, defining 
$W_{i+} := \frac{W_{i1} + W_{i2}}{\sqrt 2}$,
$W_{i-} := \frac{W_{i1} - W_{i2}}{\sqrt 2}$,
$Z_{j+} := \frac{Z_{j1} + Z_{j2}}{\sqrt 2}$, and $Z_{j-} := \frac{Z_{j1} - Z_{j2}}{\sqrt 2}$, we have
\begin{align}\label{eq: rotation1}
    h(X_{i1}+\sqrt \delta W_{i1}, X_{i2}+\sqrt \delta W_{i2} | U) = h(X_{i+}+\sqrt \delta W_{i+}, X_{i-}+\sqrt \delta W_{i-} | U ),
\end{align}
and
\begin{align}\label{eq: rotation2}
    &h(A_j(X_1 +\sqrt \delta W_1) + \sqrt \epsilon Z_{j1}, A_j(X_2 +\sqrt \delta W_2) + \sqrt \epsilon Z_{j2} | U) \nonumber \\
    &= h(A_j(X_+ +\sqrt \delta W_+) + \sqrt \epsilon Z_{j+}, A_j(X_- +\sqrt \delta W_-) + \sqrt \epsilon Z_{j-} | U).
\end{align}
$(W_1, W_2, Z_1, Z_2)$ and $(W_+, W_-, Z_+, Z_-)$ are equal in distribution. Multiplying the equations in \eqref{eq: rotation1} by $d_i$ and those in \eqref{eq: rotation2} by $c_j$ and subtracting the sum of the latter from the sum of the former, we may conclude that $s_{\epsilon, \delta}(X_1, X_2 | U) = s_{\epsilon, \delta}(X_+, X_- | U)$.
\end{proof}

\begin{lemma}\label{lemma: gaussian_delta}
Fix $\epsilon, \delta > 0$.
Let the random variables $X^*$ and $U^*$ be as in Lemma \ref{lemma: xstar_ustar}; i.e., satisfying the equality $V(\Sigma) = s_{\epsilon, \delta}(X^* | U^*)$, and with $|\cU^*| \leq M$. Consider two independent and identically distributed copies of $(X^*,U^*)$, denoted by $(X_1, U_1)$ and $(X_2, U_2)$. Define new random variables $X_+$ and $X_-$ as follows:
\begin{align*}
X_+ := \frac{X_1 + X_2}{\sqrt 2}, \quad &\text{ and } \quad X_- := \frac{X_1 - X_2}{\sqrt 2}.
\end{align*}
Also, define $U \defn (U_1, U_2)$. Then the following results hold:
\begin{enumerate}
    \item [(a)] $X_+$ and $X_-$ are conditionally independent given $U$,
    \item [(b)] $V(\Sigma) = s_{\epsilon, \delta}(X_+|U)$ and $V(\Sigma) = s_{\epsilon, \delta}(X_-|U)$.
\end{enumerate}
\end{lemma}
\begin{proof}
We have the following sequence of inequalities:
\begin{align*}
    2V(\Sigma) &\stackrel{(a)}= s_{\epsilon, \delta}(X_1 | U_1) + s_{\epsilon, \delta}(X_2 | U_2)\\
    &\stackrel{(b)}= s_{\epsilon, \delta}(X_1, X_2 | U_1, U_2)\\
    &\stackrel{(c)}= s_{\epsilon, \delta}(X_+, X_- | U_1, U_2)\\
    &\stackrel{(d)}\leq S_{\epsilon, \delta}(X_+, X_-)\\
    &\stackrel{(e)}\leq S_{\epsilon, \delta}(X_+) + S_{\epsilon, \delta}(X_-)\\
    &\stackrel{(f)}\leq V(\Sigma) + V(\Sigma) = 2V(\Sigma).
\end{align*}
Here $(a)$ follows from the assumption that $s_{\epsilon,\delta}(X^*| U^*) = V(\Sigma)$. Equality $(b)$ follows from the independence $(X_1, U_1) \ind (X_2, U_2)$. Equality $(c)$ holds because of Lemma \ref{lemma: rotation}. Inequality $(d)$ follows from the definition of $S_{\epsilon, \delta}(\cdot)$. Inequality $(e)$ follows from the tensorization result in Lemma \ref{lemma: EPI+BL+sub}. Finally, inequality $(f)$ follows from the definition in equation \eqref{eq: vV-V}, and the fact that $X_+$ and $X_-$ have the same covariance as $X^*$, which is bounded above by $\Sigma$ in the positive semidefinite partial order.

Since the first and last expressions match, all the inequalities in the above sequence of inequalities must be equalities. In particular, equalities $(d)$ and $(e)$ combined with Lemma \ref{lemma: independence} imply that $X_+ \ind X_-$ conditioned on $(U_1, U_2)$, thus establishing part  (a) of the lemma. Lemma \ref{lemma: independence} also gives $s_{\epsilon, \delta}(X_+|U_1, U_2) = S_{\epsilon, \delta}(X_+)$ and $s_{\epsilon, \delta}(X_-|U_1, U_2) = S_{\epsilon, \delta}(X_-)$. Finally, equality in $(f)$ gives $S_{\epsilon, \delta}(X_+) = V(\Sigma)$ and $S_{\epsilon, \delta}(X_-) = V(\Sigma)$. This completes the proof of part (b).
\end{proof}

\begin{lemma}\label{lemma: single_gaussian}
There exists $G^* \sim \cN(0, \Sigma^*) \in \cP(\mathbf r)$ such that $\Sigma^* \preceq \Sigma$ and $V(\Sigma) = s_{\epsilon, \delta}(G^*)$. Furthermore, the random variable $G^*$ is the unique element of the set $\cP(\mathbf r) \cap \{X ~:~ \E XX^T \preceq \Sigma\}$ satisfying $s_{\epsilon, \delta}(X) = V(\Sigma)$.
\end{lemma}
\begin{proof}
Consider the setting as in Lemma \ref{lemma: gaussian_delta}. Using Lemma \ref{lemma: gaussian_delta}, we have that $X_+ \ind X_-$ conditioned on $U = (u_1, u_2)$ for any $u_1, u_2 \in \cU^*$. However, we also have $X_1 \ind X_2$ conditioned on $U = (u_1, u_2)$. The characterization theorem for Gaussian distributions \cite{GhuOlk62} implies that  $X_1$ and $X_2$ must be Gaussian with identical covariance matrices, conditioned on $U = (u_1, u_2)$. Recall that $(X_1, U_1)$ is independent of $(X_2, U_2)$, and the covariance matrix of $X_i$ conditioned on $U= (u_1, u_2)$ is simply the covariance matrix of $X_i$ conditioned on $U_i = u_i$ for $i \in \{1, 2\}$. Since $u_1$ and $u_2$ may be chosen arbitrarily, we conclude that the covariance matrix of $X_1$ is some fixed $\Sigma^* \preceq \Sigma$ for all $u_1 \in \cU^*$. Let $G^* \sim \cN(0, \Sigma^*)$. Thus,
\begin{align*}
    V(\Sigma) &= \sum_{u_1 \in \cU^*} p_{U_1}(u_1)s_{\epsilon, \delta}(X_1 | U_1 = u_1)\\
    &=  \sum_{u_1 \in \cU^*} p_{U_1}(u_1)s_{\epsilon, \delta}(G^*)\\
    &= s_{\epsilon, \delta}(G^*).
\end{align*}
To establish uniqueness, first note that it is enough to only consider Gaussian random variables $X$ satisfying $s_{\epsilon, \delta}(X) = V(\Sigma)$, since our argument above shows that any $X$ that achieves this equality must be Gaussian. Now suppose that $G_1 \sim \cN(0, \Sigma_1)$ and $G_2 \sim \cN(0, \Sigma_2)$ are two distinct random variables such that  $s_{\epsilon, \delta}(G_1) = s_{\epsilon, \delta}(G_2) = V(\Sigma)$ with $\Sigma_1, \Sigma_2 \preceq \Sigma$. Define $(X, U)$ such that $X = G_1$ when $U=1$ and $X = G_2$ when $U=2$. Suppose also that $U$ takes values 1 and 2 with probability $1/2$, each. It is easy to check that $X$ satisfies the covariance constraint, and that $s_{\epsilon, \delta}(X|U) = V(\Sigma).$ As in Lemma \ref{lemma: gaussian_delta}, consider two i.i.d.\ copies of $(X_1,U_1)$ and $(X_2, U_2)$ of $(X, U)$. Lemma \ref{lemma: gaussian_delta} states that conditioned on $(U_1 = u_1, U_2 = u_2)$, we have $X_1+X_2 \ind X_1 - X_2$, for any values of $u_1$ and $u_2$. Conditioned on $u_1 = 1$ and $u_2 = 2$, we have $X_1 + X_2 = G_1 + G_2$ and $X_1 - X_2 = G_1 - G_2$. This implies $G_1+G_2 \ind G_1-G_2$, which is impossible since $\Sigma_1 \neq \Sigma_2$, and thus there cannot be two distinct Gaussian maximizers.  
\end{proof}

\begin{proof}[Proof of Theorem \ref{thm: EPI+BL}]
We now complete the proof of Theorem \ref{thm: EPI+BL}. Recall the definition of $M_g$:
\begin{equation*}
M_g := \sup_{X \in \cP_g(\mathbf r)} \sum_{i=1}^k d_i h(X_{i}) - \sum_{j=1}^m c_j h(A_j X).
\end{equation*}
Clearly, there is nothing to prove if $M_g$ is infinite, so we assume $M_g < \infty$.  Let $X \in \cP(\mathbf r)$ be an arbitrary random vector. By choosing a large enough $\Sigma$ such that $\E XX^T \preceq \Sigma$, we may conclude that
\begin{equation}\label{eq: sV}
s_{\epsilon, \delta}(X) \leq V(\Sigma).
\end{equation}
Let $G^*  \sim \cN(0, \Sigma^*) \in \cP(\mathbf r)$, where $\Sigma^* \preceq \Sigma$, be the unique maximizer such that $s_{\epsilon, \delta}(G^*) = V(\Sigma)$, 
as in Lemma \ref{lemma: single_gaussian}. Thus, we have the sequence of inequalities
\begin{align}
    V(\Sigma) &= \sum_{i=1}^k d_i h(G^*_i + \sqrt \delta W_i) - \sum_{j=1}^m c_j h(A_j(G^* + \sqrt \delta W) + \sqrt \epsilon Z_j) \nonumber\\
    &\stackrel{(a)}\leq \sum_{i=1}^k d_i h(G^*_i + \sqrt \delta W_i) - \sum_{j=1}^m c_j h(A_j(G^* + \sqrt \delta W)) \nonumber\\
    &\stackrel{(b)}\leq M_g.  \label{eq: VMg}
\end{align}
Here, inequality $(a)$ follows from the entropy inequality 
\begin{align*}
h(A_j (G^*+\sqrt \delta W) + \sqrt \epsilon Z_j) \geq h(A_j (G^*+\sqrt \delta W)),
\end{align*}
for all $j \in [m]$. The inequality in $(b)$ is true because the random variable $\tilde {G^*}$ defined by $\tilde {G^*_i} \defn G^*_i + \sqrt \delta W_i$ for $i \in [k]$ is a Gaussian random variable in $\cP_g(\mathbf r)$. Thus, by the definition of $M_g$, we must have
\begin{align*}
\sum_{i=1}^k d_i h(\tilde {G^*_i} ) - \sum_{j =1}^m c_j h(A_j \tilde {G^*} ) &\leq \sup_{X \in \cP_g(\mathbf r)} \sum_{i=1}^k d_i h(X_{i}) - \sum_{j=1}^m c_j h(A_j X)\\
&= M_g.
\end{align*}
Combining inequalities \eqref{eq: sV} and \eqref{eq: VMg}, we have
\begin{equation}\label{eq: sMg}
s_{\epsilon, \delta}(X) \leq M_g.
\end{equation}
Recall that $s_{\epsilon, \delta}(X)$ is given by
\begin{align*}
s_{\epsilon, \delta}(X) = \sum_{i=1}^k d_i h(X_i + \sqrt \delta W_i) - \sum_{j =1}^m c_j h(A_j (X +\sqrt \delta W) + \sqrt \epsilon Z_j).
\end{align*} 
If $X$ satisfies certain  mild conditions (such as bounded second moments) provided in  Lemma \ref{lemma: heat}, we have that
\begin{align*}
\lim_{\epsilon, \delta \to 0} s_{\epsilon, \delta}(X) =  \sum_{i=1}^k d_i h(X_i) - \sum_{j =1}^m c_j h(A_j X).
\end{align*}
This means that we may take the limit in inequality \eqref{eq: sMg} as $\epsilon, \delta \to 0$ to conclude 
\begin{align*}
 \sum_{i=1}^k d_i h(X_i) - \sum_{j =1}^m c_j h(A_j X) \leq M_g,
\end{align*}
and conclude the proof of Theorem \ref{thm: EPI+BL}.
\end{proof}

\section{Conditions for $M(\mathbf A, \mathbf c, \mathbf r, \mathbf d) < \infty$}\label{section: finite}

Theorem \ref{thm: EPI+BL} shows that it is enough to find necessary and sufficient conditions for $M_g(\mathbf A, \mathbf c, \mathbf r, \mathbf d)$ to be finite, since $M = M_g$. We prove Theorem \ref{thm: finite} by finding necessary conditions on the BL-EPI datum
for such finiteness
in Claim \ref{claim: necessary_finite}, and showing that the necessary conditions are also sufficient in Claim \ref{claim: sufficient_finite}. 

\begin{claim}\label{claim: necessary_finite}
If $M_g(\mathbf A, \mathbf c, \mathbf r, \mathbf d)$ is finite, then the conditions in equations \eqref{eq: finite1} and \eqref{eq: finite2} must be satisfied.
\end{claim}
\begin{proof}
The necessity of the condition in equation \eqref{eq: finite2} is seen as follows. Choose $Z \sim \lambda \cN(0, I_{n\times n})$ for some $\lambda > 0$. It is easy to see that $\sum_{i=1}^k d_i h(Z_i) - \sum_{j=1}^m c_j h(A_jZ)$ scales as $\left(\sum_{i=1}^k d_i r_i - \sum_{j=1}^m c_j n_j\right) \log(\lambda)$ as a function of $\lambda$ as $\lambda \to \infty$. Since $\lambda$ is arbitrary, the above expression is finite only if the condition in equation \eqref{eq: finite2} is satisfied.

To show that the condition in equation \eqref{eq: finite1} is necessary, let $V$ be a subspace of $\real^n$ of $\mathbf r$-product form. Consider a Gaussian random variable $Z \defn (Z_V, Z_{V^\perp})$ such that $Z_V \ind Z_{V^\perp}$, and $Z_V$ is supported on $V$ and $Z_{V^\perp}$ is supported on $V^\perp$. Furthermore, assume $Z_V \sim \cN(0, \lambda I_{\dim(V) \times \dim(V)})$ and $Z_{V^\perp} \sim \cN(0, I_{\dim(V^{\perp} \times \dim(V^\perp)}))$. Taking the limit as $\lambda\to \infty$ and gathering the coefficients of $\log \lambda$, we see that $M_g(\mathbf A, \mathbf c, \mathbf r, \mathbf d)$ scales as
$$\left(\sum_{i=1}^k d_i \dim(V_i) - \sum_{j=1}^k c_j \dim(A_j V)\right) \log(\lambda),$$
as $\lambda \to \infty$. Thus, $M_g$ is finite only if the condition in equation \eqref{eq: finite1} is satisfied.
\end{proof}

The proof of sufficiency of the conditions in equations \eqref{eq: finite1} and \eqref{eq: finite2} relies on two lemmas which we prove below.
\begin{lemma}\label{lemma: split}
Let $(\mathbf A, \mathbf c, \mathbf r, \mathbf d)$ be a BL-EPI datum. Let $U := (U_1, \dots, U_k)$ be an arbitrary $\mathbf r$-product form subspace such that $\dim(U_i) = \tilde r_i \leq r_i$ for $i \in [k]$. Let $\tilde{\mathbf r} := (\tilde r_1, \dots, \tilde r_k)$ and $\tilde{\mathbf r}^c := \mathbf r - \tilde{\mathbf r}$. Define two BL-EPI data as follows:
\begin{enumerate}
\item[(a)] $(\tilde{\mathbf A}, \mathbf c, \tilde{\mathbf r}, \mathbf d)$ is a BL-EPI datum defined on $U$. For each $j \in [m]$, define the linear maps $\tilde A_j: U \to (A_jU)$  by $\tilde A_jx = A_j x$ for $x \in U$. 
\item[(b)] $(\tilde{\tilde{\mathbf A}}, \mathbf c, \tilde{\mathbf r}^c, \mathbf d)$ is a BL-EPI datum defined on $U^\perp$. For $j \in [m]$, the linear maps $\tilde{\tilde{A}}_j: U^\perp \to (A_jU)^\perp$ are defined by
\begin{align*}
\tilde{\tilde{A}}_jx = \Pi_{(A_jU)^\perp} A_j x.
\end{align*}
We also define the linear maps $\Gamma_j: U^\perp \to (A_jU)$ as
\begin{align*}
\Gamma_j x = \Pi_{(A_jU)} A_j x.
\end{align*}
Here $\Pi_V$ denotes the orthogonal projection on to a subspace $V$. Note that $A_j x = \tilde{\tilde{\mathbf A}}x + \Gamma_j x$ is an orthogonal decomposition.
\end{enumerate}
Then the following relation holds:
\begin{equation}
M(\mathbf A, \mathbf c, \mathbf r, \mathbf d) \leq M(\tilde{\mathbf A}, \mathbf c, \tilde{\mathbf r}, \mathbf d) + M(\tilde{\tilde{\mathbf A}}, \mathbf c, \tilde{\mathbf r}^c, \mathbf d).
\end{equation}
\end{lemma}
\begin{remark}
Note that it may happen that $\dim(U_i) = 0$ for some $i \in [k]$. It may also happen that for some $j \in [m]$, we have $\dim((A_jU)^\perp) = 0$. We do not rule out such cases, and keep our notation the same by instead defining entropy on a 0-dimensional subspace as 0.
\end{remark}

\begin{proof}
[Proof of Lemma \ref{lemma: split}]
By definition, the linear transformations in $\tilde{\mathbf A}$ and $\tilde{\tilde{\mathbf A}}$ are surjective. Also, $\sum_{i} \tilde r_i = \dim(U)$ and $\sum_{i} \tilde r^c_i = \dim(U^\perp)$. This verifies that $(\tilde{\mathbf A}, \mathbf c, \tilde{\mathbf r}, \mathbf d)$ and $(\tilde{\tilde{\mathbf A}}, \mathbf c, \tilde{\mathbf r}^c, \mathbf d)$ are indeed valid BL-EPI data on $U$ and $U^\perp$, respectively. Every vector $x \in \real^n$ may be expressed as $x = \Pi_U x + \Pi_{U^\perp} x \defn \tilde x + \tilde {\tilde x}$. We use the notation $\tilde x = (\tilde x_1, \dots, \tilde x_k)$ where $\tilde x_i = \Pi_{U_i} x_i$, and similarly for $\tilde{\tilde x}_i$. We have the equality
\begin{align*}
A_j x &= A_j( \Pi_U x + \Pi_{U^\perp} x)\\
&= A_j(\Pi_U x) + A_j(\Pi_{U^\perp} x)\\
&= \tilde A_j \tilde x + \Pi_{(A_jU)} A_j \tilde{\tilde x} + \Pi_{(A_jU)^\perp} A_j \tilde{\tilde x}\\
&= \tilde A_j \tilde x + \Gamma_j \tilde{\tilde x} + {\tilde{\tilde A}}_j {\tilde{\tilde x}}.
\end{align*}
For any $X \in \cP(\mathbf r)$, 
\begin{align*}
\sum_{i=1}^k d_i h(X_i) - \sum_{j=1}^m c_j h(A_j X) &= \sum_{i=1}^k d_i h(X_i) - \sum_{j=1}^m c_j h( \tilde A_j \tilde X 
+ \Gamma_j \tilde{\tilde X} 
+ {\tilde{\tilde A}}_j \tilde{\tilde X})\\
&= \sum_{i=1}^k d_i h(\tilde X_i, \tilde{\tilde X}_i) - \sum_{j=1}^m c_j h( \tilde A_j \tilde X + \Gamma_j \tilde{\tilde X}, {\tilde{\tilde A}}_j {\tilde{\tilde X}})\\
&= \sum_{i=1}^k d_i h(\tilde{\tilde X}_i) - \sum_{j=1}^m c_j h({\tilde{\tilde A}}_j {\tilde{\tilde X}}) + \sum_{i=1}^k d_i h(\tilde X_i \mid \tilde{\tilde X}_i) - \sum_{j=1}^m c_j h( \tilde A_j \tilde X + \Gamma_j \tilde{\tilde X} \mid {\tilde{\tilde A}}_j {\tilde{\tilde X}})\\
&\leq \sum_{i=1}^k d_i h(\tilde{\tilde X}_i) - \sum_{j=1}^m c_j h({\tilde{\tilde A}}_j {\tilde{\tilde X}}) + \sum_{i=1}^k d_i h(\tilde X_i \mid \tilde{\tilde X}_i) - \sum_{j=1}^m c_j h( \tilde A_j \tilde X \mid {\tilde{\tilde X}})\\
&= \sum_{i=1}^k d_i h(\tilde{\tilde X}_i) - \sum_{j=1}^m c_j h({\tilde{\tilde A}}_j {\tilde{\tilde X}}) + \sum_{i=1}^k d_i h(\tilde X_i \mid \tilde{\tilde X}) - \sum_{j=1}^m c_j h( \tilde A_j \tilde X \mid {\tilde{\tilde X}})\\
&\leq M(\tilde{\mathbf A}, \mathbf c, \tilde{\mathbf r}, \mathbf d) + M(\tilde{\tilde{\mathbf A}}, \mathbf c, \tilde{\mathbf r}^c, \mathbf d).
\end{align*}
Taking the supremum over all $X \in \cP(\mathbf r)$ completes the proof.
\end{proof}

\begin{lemma}\label{lemma: split2}
Suppose that a BL-EPI datum $(\mathbf A, \mathbf c, \mathbf r, \mathbf d)$  satisfies the conditions in equations \eqref{eq: finite1} and \eqref{eq: finite2}, and suppose that $U$ is an $\mathbf r$-product form critical subspace. Then the BL-EPI data $(\tilde{\mathbf A}, \mathbf c, \tilde{\mathbf r}, \mathbf d)$ and $(\tilde{\tilde{\mathbf A}}, \mathbf c, \tilde{\mathbf r}^c, \mathbf d)$ defined as in Lemma \ref{lemma: split} also satisfy the conditions in equations \eqref{eq: finite1} and \eqref{eq: finite2}.
\end{lemma}
\begin{proof}
Verifying the conditions for $(\tilde{\mathbf A}, \mathbf c, \tilde{\mathbf r}, \mathbf d)$ is immediate: the condition in equation \eqref{eq: finite1} restricted to $\tilde{\mathbf r}$ product form subspaces of $U$ yields the first condition, and the criticality of $U$ yields the second condition.

For $j \in [m]$, it is not hard to verify that $\dim(\tilde{\tilde A}_j U^\perp)$ is $n_j - \dim(\tilde A_jU)$. We may now check the second condition for  $(\tilde{\tilde{\mathbf A}}, \mathbf c, \tilde{\mathbf r}^c, \mathbf d)$ by observing the equality
\begin{align*}
\sum_{i=1}^k d_i(r_i - \dim(U_i)) = \sum_{j=1}^m c_j(n_j - \dim({\tilde A}_j U)),
\end{align*}
using the criticality of $U$ and the fact that $\sum_{i=1}^k d_i r_i = \sum_{j=1}^m c_j n_j$. Let $V$ be an arbitrary $\tilde{\mathbf r}^c$-product form subspace of $U^\perp$. Consider the new subspace $V_+ = V \oplus U \subset \real^n$, which is the direct sum of the subspace $V$ with the subspace $U$.
Note that $V_+$ is an $\mathbf r$-product form subspace of
$\real^n$. Using the condition in equation \eqref{eq: finite1} for $V_+$, we have
\begin{align*}
\sum_{i=1}^k d_i \dim(V_{+i}) \le \sum_{j=1}^m c_j \dim(A_j V_+).
\end{align*}
Note that $\dim(V_{+i}) =\dim(V_i) + \dim(U_i)$, for all 
$1 \le i \le k$. Moreover, $\dim(A_j V_+) = \dim(A_j U)+ \dim(\tilde{\tilde A}_j V_i)$. Substituting these equalities in the above inequality, we arrive at
\begin{align*}
\sum_{i=1}^k d_i (\dim(V_i) + \dim(U_i)) \le \sum_{j=1}^m c_j (\dim(A_j U)+ \dim(\tilde{\tilde A}_j V_i)).
\end{align*}
The criticality of $U$ then implies
\begin{align*}
\sum_{i=1}^k d_i \dim(V_i)  \le \sum_{j=1}^m c_j  \dim(\tilde{\tilde A}_j V_i),
\end{align*}
and this completes the proof.
\end{proof}
We are now in a position to prove the following sufficiency result:
\begin{claim}\label{claim: sufficient_finite}
If the conditions in equations \eqref{eq: finite1} and \eqref{eq: finite2} are satisfied, then $M(\mathbf A, \mathbf c, \mathbf r, \mathbf d)$ is finite.	
\end{claim}
\begin{proof}
The proof proceeds via a double induction on the dimension $n$ and the number of linear maps $m$. We first prove the result for $n=1$ and arbitrary $m$, and for $m=1$ and arbitrary $n$. For $n=1$, it must be that $\mathbf r = \{1\}$ and $\mathbf d = \{d_1\}$. The conditions in equations \eqref{eq: finite1} and \eqref{eq: finite2} imply that $d_1 = \sum_{j=1, n_j > 0}^m c_j$, because $n_j > 0 \Longrightarrow n_j =1$. Thus, $M(\mathbf A, \mathbf c, \mathbf r, \mathbf d)$ equals 
\begin{align*}
    \sup_{X \in \cP(\mathbf r)} d_1h(X) - \sum_{j=1}^m c_j h(A_j X) 
    &= \sup_{X \in \cP(\mathbf r)} d_1h(X) - \sum_{j=1, n_j > 0}^m c_j h(A_j X)\\
    &= \sup_{X \in \cP(\mathbf r)} -\sum_{j=1, n_j > 0}^m c_j \log|A_j|\\
    &= -\sum_{j=1, n_j > 0}^m c_j \log|A_j| < \infty,
\end{align*}
since $h(A_jX) = h(X) + \log |A_j|$ for all $j \in [m]$ such that $n_j > 0$, and $A_j$ is a nonzero scalar for each such $j$.

Now fix $m=1$ and let $n_1 > 0$, $k>0$, $\mathbf r$, $\mathbf d$, $c_1$, and $n = \sum_{i=1}^k r_i$ be arbitrary,
subject to satisfying the conditions in equations \eqref{eq: finite1} and \eqref{eq: finite2}.
We write
\[
A_1 = \left[ A_{11} \ldots A_{1k} \right]
\]
where $A_{1i}$ is an $n_1 \times r_i$ matrix for 
$1 \le i \le k$
(and $A$ is an $n_1 \times n$ matrix).
Recall that, by assumption, $d_i > 0$ for $1 \le i \le k$
and $c_1 > 0$.

Let $\mcN(A_1)$ denote the null space of $A_1$. 
For every $\mbfr$-product form subspace 
$V := V_1 \times \ldots \times V_k$
we must have $\mcN(A_1) \cap V_i = \{ 0\}$ for all 
$1 \le i \le k$. This is because if we have 
$0 \neq v_i \in \mcN(A_1) \cap V_i$ for some
$1 \le i \le k$, then letting $V_i := \mbox{span}(\{v_i\})$
and $V_j = \{0\}$ for $1\le j \neq i \le k$, the 
corresponding $\mbfr$-product form subspace 
$V := V_1 \times \ldots \times V_k$ will 
violate the condition
$\sum_{i=1}^k d_i s_i \le \mdim(A_1 V)$, where 
$s_i = 1 = \mdim(V_i)$ and 
$s_j = 0 = \mdim(V_j)$ for $1\le j \neq i \le k$.

We can therefore assume that $\mbox{rk}(A_{1i}) = r_i$
for $1 \le i \le k$. Under this assumption, we will now show that 
$M_g < \infty$, where $M_g$ denotes the supremum of 
\[
\sum_{i=1}^k d_i h(X_i) - c_1 h(A_1 X)
\]
over independent $X_i \sim \mcN(0, \Sigma_i)$ taking values in 
$\mbbR^{r_i}$ with $\Sigma_i$ positive definite
for each  for $1 \le i \le k$, and where 
\[
X := \left[ X_1 \ldots X_k \right]^T.
\]

We have 
\[
h(X_i) = \frac{1}{2} \log \left( (2 \pi e)^{r_i}~ \mbox{det}(\Sigma_i) \right),
\]
for $1 \le i \le k$, and 
\[
h(X) = \frac{1}{2} \log \left( (2 \pi e)^{n_1}~ \mbox{det}(\sum_{i=1}^k A_{1i} \Sigma_i A_{1i}^T) \right).
\]
It is therefore equivalent to show that the supremum of 
\[
\sum_{i=1}^k d_i \log \left(  \mbox{det}(\Sigma_i)  \right)
- c_1 \log \left( \mbox{det}(\sum_{i=1}^k A_{1i} \Sigma_i A_{1i}^T) \right),
\]
over $\Sigma_i \in \mbbR^{r_i \times r_i}$ positive definite
for each  for $1 \le i \le k$ is finite.

Let $A_{1i} = W_i \Lambda_i U_i^T$ be a singular value decomposition of $A_{1i}$ for $1 \le i \le k$.
Since $\mbox{rk}(A_{1i}) = r_i$ by assumption, here $\Lambda_i$ is a diagonal $r_i \times r_i$ matrix
with strictly positive diagonal entries, $U_i$ is an $r_i \times r_i$ orthogonal matrix and
$W_i$ is an $n_1 \times r_i$ matrix with orthonormal columns. Note that span of the columns of 
$W_i$ equals the range space of $A_{1i}$.

With $\tilSigma_i$ denoting $U_i \Sigma_i U_i^T$ for $1 \le i \le k$, it is equivalent to show that 
the supremum of 
\[
\sum_{i=1}^k d_i \log \left(  \mbox{det}(\tilSigma_i)  \right)
- c_1 \log \left( \mbox{det}(\sum_{i=1}^k W_i \Lambda_i \tilSigma_i \Lambda_i W_i^T) \right),
\]
over $\tilSigma_i \in \mbbR^{r_i \times r_i}$ positive definite
for each  for $1 \le i \le k$ is finite.

Note that the entries of $\Lambda_i$ depend only on $A_{1i}$, which is fixed, and note that the 
$d_i$ are fixed. Therefore, with 
$\hatSigma_i$ denoting $\Lambda_i \tilSigma_i \Lambda_i$, it is equivalent to show that 
the supremum of 
\[
\sum_{i=1}^k d_i \log \left(  \mbox{det}(\hatSigma_i)  \right)
- c_1 \log \left( \mbox{det}(\sum_{i=1}^k W_i \hatSigma_i W_i^T) \right),
\]
over $\hatSigma_i \in \mbbR^{r_i \times r_i}$ positive definite
for each  for $1 \leq i \leq k$ is finite. Let $\hatSigma_i = \hat{Q}_i \Pi_i \hat{Q}_i^T$ be the spectral-decomposition of $\hatSigma_i$ and 
let $\sigma_{i1}, \ldots, \sigma_{1 r_i}$ denote the eigenvalues of $\hatSigma_i$ in any order.
By assumption these are all strictly positive. Let
\[
\sigma_1 > \sigma_2 > \ldots > \sigma_{n'} > 0 =: \sigma_{n'+1},
\]
denote the ordered list of all the distinct values among these eigenvalues (note that $n = \sum_{i=1}^k r_i$,
so here $1 \le n' \le n$).

Starting with $\sigma_{n'}$ and working towards the larger eigenvalues step by by step we can build up each
$\hatSigma_i$, for $1 \le i \le k$, in layer-cake fashion as
\[
\hatSigma_i = \hatSigma_{i1} + \hatSigma_{i2} + \ldots + \hatSigma_{in'}
\]
where each $\hatSigma_{il}$  for $1 \le l \le n'$ is a positive semidefinite matrix, with a spectral decomposition given by $\hat{Q}_i \Pi_{il} \hat{Q}_i^T$, and each of whose eigenvalues is either
$0$ or $\sigma_j - \sigma_{j+1}$ (recalling the convention that $\sigma_{n'+1} = 0$). 
Thus each $\hatSigma_{il}$ corresponds to a subspace of $\mbbR^{r_i}$, whose dimension we denote as 
$s_{il}$. Note that $s_{in'} = r_i$ and $s_{il}$ is nonincreasing as $l$ decreases,
but it can become $0$ for $l < n'$; however we have $s_{i1} > 0$ for at least one choice of $1 \le i \le k$.
We also have
\[
\sum_{i=1}^k d_i \log \left(  \mbox{det}(\hatSigma_i)  \right) = 
\sum_{l=1}^{n'-1} \left( \sum_{i=1}^k d_i s_{il} \log \frac{\sigma_l}{\sigma_{l+1}}  \right)
+ \sum_{i=1}^k d_i r_i \log \sigma_{n'}.
\]
Observe that $\log \frac{\sigma_l}{\sigma_{l+1}}$ is strictly positive for $1 \le l \le n' -1$.

Let $\hatV_{il}$ denote the subspace of $\mbbR^{r_i}$ corresponding to $\hatSigma_{il}$, i.e. the subspace spanned by the eigenvectors of  $\hatSigma_{il}$. Then 
$\tilV_{il} := \Lambda_i^{-1} \hatV_{il}$ is the subspace corresponding to $\tilSigma_{il}$ in the 
same sense, where $\tilSigma_{il} := \Lambda_i^{-1} \hatSigma_{il} \Lambda_i^{-1}$, and 
$V_{il} := U_i^T \tilV_{il}$ is the subspace corresponding to $\Sigma_{il}$ in the 
same sense, where $\Sigma_{il} := U_i^T \tilSigma_{il} U_i$. Note that
\[
\mdim(V_{il}) = \mdim(\tilV_{il}) = \mdim(\hatV_{il}) =: s_{il}.
\]
By assumption, for each $1 \le l \le n'$ we therefore have
\[
\sum_{i=1}^k d_i s_{il} \le c_1 \mdim(A_1 V_l),
\]
where $V_l := V_{1l} \times \ldots \times V_{kl}$ is an $\mbfr$-product subspace of $\mbbR^n$.

For each $1 \le l \le n'$, since $\sum_{i=1}^k A_{1i} \Sigma_{il} A_{1i}^T = \sum_{i=1}^k W_i \hatSigma_{il} W_i^T$,
we see that the subspace corresponding to $\sum_{i=1}^k W_i \hatSigma_{il} W_i^T$ is $A_1 V_l$. 
In particular, the subspace corresponding to $\sum_{i=1}^k W_i \hatSigma_{in'} W_i^T$ is $\mbbR^{n_1} = A_1 V_{n'} = A_1 \mbbR^n$.

We also note that for each $1 \le i \le k$ we have
\[
\hatV_{i1} \subseteq \hatV_{i2} \subseteq \ldots \subseteq \hatV_{in'} = \mbbR^{r_i}.
\]
Since $\hatSigma_i = \hat{Q}_i \Pi_i \hat{Q}_i^T = \sum_{m=1}^{r_i} \sigma_{im}\hat{q}_{im}\hat{q}_{im}^T$, let us relabel the eigenvectors into $b_{im}$ (according to decreasing values of the eigenvalues) such that we have
\[
\hatSigma_{il} = (\sigma_l - \sigma_{l+1}) \sum_{u_i =1}^{s_{il}} b_{i u_i} b_{i u_i}^T,
\]
where we recall that $\sigma_{n'+1} = 0$ by definition.
We can also write
\[
W_i \hatSigma_{il} W_i^T = (\sigma_l - \sigma_{l+1}) \sum_{u_i =1}^{s_{il}} W_i b_{i u_i} b_{i u_i}^T W_i^T
= (\sigma_l - \sigma_{l+1}) \sum_{u_i =1}^{s_{il}} \tilb_{i u_i} \tilb_{i u_i}^T,
\]
where $\tilb_{i u_i} := W_i b_{i u_i}$ for $1 \le i \le k$ and $1 \le u_i \le r_i$. Note that 
$\tilb_{i u_i} \in \mbbR^{n_1}$. 

Now we have
\begin{eqnarray*}
\sum_{i=1}^k W_i \hatSigma_i W_i^T &=& \sum_{i=1}^k \sum_{l=1}^{n'} W_i \hatSigma_{il} W_i^T \\
&=& \sum_{l=1}^{n'} (\sigma_l - \sigma_{l+1}) \sum_{i=1}^k \sum_{u_i =1}^{s_{il}} \tilb_{i u_i} \tilb_{i u_i}^T\\
&=& \sum_{l=1}^{n'} (\sigma_l - \sigma_{l+1}) M_l,
\end{eqnarray*}
where $M_l := \sum_{i=1}^k \sum_{u_i =1}^{s_{il}} \tilb_{i u_i} \tilb_{i u_i}^T$. 
Note that the subspace corresponding to $M_l$ is $A_1 V_l$. Since the range space of $M_l$ is non-decreasing,  there exists an orthonormal basis $\tilde{q}_1,...,\tilde{q}_{n_1}$ for $\mathbb{R}^{n_1}$ such that the range space of $M_l$ matches the span of $\{q_i\}_{i \in S_l}$ for some appropriate $S_l \subseteq [1:n_1]$. Thus $\textrm{dim}(A_1V_l) = |S_l|$.

By construction we have $S_1 \subseteq S_2 \subseteq \cdots \subseteq S_{n'}=[1:n_1]$. Let $C_l = \sum_{i\in S_l} \tilde{q}_i \tilde{q}_i^T = \tilde{Q} \Theta_l \tilde{Q}^T$ where $\tilde{Q}$ is the orthonormal matrix formed by $\tilde{q}$'s and $\Theta_l$ is a diagonal matrix with diagonal entries being $0$ or $1$, where $1$ occurs at the indices corresponding to the membership in $S_l$.

We now claim that there is positive constant $\delta^2 > 0$ depending only on 
$W_1, \ldots, W_k$ (and in particular not depending on the 
$(\hatSigma_i, 1 \le i \le k)$ or the choices of the bases
$\{b_{i1}, b_{i2}, \ldots, b_{i r_i} \}$ for $1 \le i \le k$)
such that, for all $1 \le l \le n'$, we have
\[
M_l \succeq \delta^2 C_l.
\]
This is a consequence of Lemma~\ref{lemma:ball}
and is established in 
Corollary~\ref{corollary:ball}.

We therefore have
\[
\sum_{i=1}^k W_i \hatSigma_i W_i^T 
\succeq \delta^2 \sum_{l=1}^{n'} (\sigma_l - \sigma_{l+1}) C_l =  \delta^2 \sum_{l=1}^{n'} (\sigma_l - \sigma_{l+1}) Q \Theta_l Q^T \succeq 0.
\]
From this it follows that 
\begin{eqnarray*}
c_1 \log \left( \mbox{det}(\sum_{i=1}^k W_i \hatSigma_i W_i^T) \right) 
&\ge&  c_1  \log \left( \mbox{det}(\sum_{l=1}^{n'}(\sigma_l - \sigma_{l+1})  \Theta_l) \right) + \kappa,\\
&\stackrel{(a)}{=}& c_1 \sum_{l=1}^{n'-1} \mdim(A V_l) \log \frac{\sigma_l}{\sigma_{l+1}}
+ c_1 n_1 \log \sigma_{n'} + \kappa,
\end{eqnarray*}
for a fixed constant $\kappa$. Here, to justify step (a),  due to the nested nature of $S_l$, $\sum_{l=1}^{n'}(\sigma_l - \sigma_{l+1})  \Theta_l $ is a diagonal matrix with $\mdim(A V_l) - \mdim(A V_{l-1}) $ entries equal to $\sigma_l$. We take $\mdim(A V_0)=0$.

Since $\sum_{i=1}^k d_i s_{il} \le c_1 \mdim(A V_l)$ and 
$\log \frac{\sigma_l}{\sigma_{l+1}}$ is strictly positive for $1 \le l \le n' -1$,
and since $\sum_{i=1}^k d_i r_i = c_1 n_1$, we can conclude that
\[
\sum_{i=1}^k d_i \log \left(  \mbox{det}(\hatSigma_i)  \right)
- c_1 \log \left( \mbox{det}(\sum_{i=1}^k W_i \hatSigma_i W_i^T) \right) \le - \kappa
\]
for all choices of $\hatSigma_i \in \mbbR^{r_i \times r_i}$ positive definite
for each  for $1 \le i \le k$. This establishes what was desired, when $m=1$.

\color{black}

We have shown that the claim is true for $n=1$ and all $m$. Assume that claim is true for all $n < n_0$ and all $m$. Our goal is to establish the claim for $n = n_0$ and all $m > 0$. To do so, we induct on $m$. The case of $n=n_0$ and $m=1$ follows from our calculations above. Now we assume that the claim is true for $n=n_0$ and all $m < m_0$, and show that it also holds for $n = n_0$ and $m = m_0$.

Let $(\mathbf A, \mathbf c, \mathbf r, \mathbf d)$ be a BL-EPI datum in $\real^{n_0}$ with $m = m_0$. We may assume that $n_j > 0$ for all $j \in [m]$, since otherwise we could have treated the scenario as a BL-EPI datum in $\real^{n_0}$ with $m < m_0$, which is already covered by the inductive hypothesis. For fixed $\mathbf A$, $\mathbf r$, and $\mathbf d$, consider the function defined on $\mathbf c \in \real_+^{m_0}$ as
\begin{equation}
M(\mathbf c) = \sup_{X \in \cP(\mathbf r)} \sum_{i=1}^k d_i h(X_i) - \sum_{j=1}^m c_j h(A_j X).
\end{equation}
Since $M$ is a pointwise supremum of linear functions, $M$ is convex. Let $\cK$ be the region of all $\mathbf c \in \real_+^{m_0}$ such that $(\mathbf A, \mathbf c, \mathbf r, \mathbf d)$ satisfy the conditions in equations \eqref{eq: finite1} and \eqref{eq: finite2}. Note that $\cK$ is a compact, convex set. By Claim \ref{claim: necessary_finite}, we have that $M$ takes $+\infty$ values outside $\cK$. We wish to show that $M$ takes finite values everywhere 
on $\cK$. Since $M$ is convex and $\cK$ is closed, it is enough to show finiteness of $M$ at all points on the boundary of $\cK$. Since $n_j > 0$ for all $j \in [m]$, a point $\mathbf c$ is a boundary point of $\cK$ if and only if at least one of the following two conditions is satisfied: (1) $c_{j_0} = 0$ for some $j_0 \in [m]$; or (2) there exists a proper $\mathbf r$-product form subspace of $\real^{n_0}$ that is critical. If a boundary point satisfies (1), then our induction assumption (on $m$) ensures the finiteness of $M$ evaluated at that BL-EPI datum,
since we could have treated the scenario as a BL-EPI datum in $\real^{n_0}$ with $m < m_0$. 

Now consider a boundary point that satisfies (2),
assuming that $c_j  \neq 0$ for all $j \in [m]$.
Let $V = (V_1, \dots, V_k)$ be an $\mathbf r$-product form critical subspace  of $\real^{n_0}$; i.e., a subspace that satisfies the equality
\begin{equation}
\sum_{i=1}^k d_i \dim(V_i) = \sum_{j=1}^m c_j \dim(A_j V),
\end{equation}
with $\dim(V) < n_0.$ Lemma \ref{lemma: split} shows that given any $\mathbf r$-product form subspace $V$, it is possible to define BL-EPI data on $V$ and $V^\perp$ in terms of the original BL-EPI datum $(\mathbf A, \mathbf c, \mathbf r, \mathbf d)$ that satisfy a certain subadditivity property. In particular, if the datum on $V$ is denoted by $(\tilde{\mathbf A}, \mathbf c, \tilde{\mathbf r}, \mathbf d)$ and that on $V^\perp$ is denoted by $(\tilde{\tilde{\mathbf A}}, \mathbf c, \tilde{\mathbf r}^c, \mathbf d)$, then Lemma \ref{lemma: split} states that
\begin{align*}
M(\mathbf A, \mathbf c, \mathbf r, \mathbf d) \le M(\tilde{\mathbf A}, \mathbf c, \tilde{\mathbf r}, \mathbf d) + M(\tilde{\tilde{\mathbf A}}, \mathbf c, \tilde{\mathbf r}^c, \mathbf d).
\end{align*}
Thus, to show that $M(\mathbf A, \mathbf c, \mathbf r, \mathbf d)$ is finite, is enough to show that $M(\tilde{\mathbf A}, \mathbf c, \tilde{\mathbf r}, \mathbf d)$ and $M(\tilde{\tilde{\mathbf A}}, \mathbf c, \tilde{\mathbf r}^c, \mathbf d)$ are finite. Lemma \ref{lemma: split2} asserts that since $V$ is a critical $\mathbf r$-product form subspace, the BL-EPI data $(\tilde{\mathbf A}, \mathbf c, \tilde{\mathbf r}, \mathbf d)$ and $(\tilde{\tilde{\mathbf A}}, \mathbf c, \tilde{\mathbf r}^c, \mathbf d)$ satisfy both the conditions in equations \eqref{eq: finite1} and \eqref{eq: finite2}. Since $\dim(V), \dim(V^\perp) < n_0$, we may use the induction assumption (on the dimension) to assert $M(\tilde{\mathbf A}, \mathbf c, \tilde{\mathbf r}, \mathbf d) < \infty$ and $M(\tilde{\tilde{\mathbf A}}, \mathbf c, \tilde{\mathbf r}^c, \mathbf d) < \infty$, and conclude the proof.
\end{proof}

\section{A special case}\label{section: special}

We examine a special case here to see what kinds of new inequalities may result from Theorem \ref{thm: EPI+BL}. Let $X_1, X_2,$ and $Y$ be real valued random variables such that $(X_1, X_2) \ind Y$. We would like to lower bound the entropy $h(X_1+Y, X_2+Y)$. Note that the regular EPI applied with the independent random vectors $(X_1, X_2)$ and $(Y, Y)$ yields the trivial lower bound
\begin{align*}
e^{h(X_1+Y, X_2+Y)} \geq e^{h(X_1, X_2)} + e^{h(Y,Y)} = e^{h(X_1, X_2)}.
\end{align*}
Note also that 
$$
\begin{pmatrix}
X_1 + Y\\
X_2 + Y
\end{pmatrix}
=
\begin{pmatrix}
1 &0 &1\\
0 &1 &1\\
\end{pmatrix}
\begin{pmatrix}
X_1\\
X_2\\
Y
\end{pmatrix}.
$$
However, it is not possible to use Zamir and Feder's EPI to provide lower bounds on $h(X_1+Y, X_2+Y)$ because of the dependency between $X_1$ and $X_2$. We show that Theorem \ref{thm: EPI+BL} may be used to obtain a family of nontrivial lower bounds that account for this dependency.

\begin{lemma}\label{lemma: abcd}
Let $\alpha, \beta, \delta_1, \delta_2 \geq 0$. Consider the inequality
\begin{equation}
\alpha h(X_1, X_2) + \beta h(Y) \leq h(X_1+Y, X_2+Y) + \delta_1 h(X_1) + \delta_2 h(X_2) + C(\alpha, \beta, \delta_1, \delta_2),
\end{equation}
where $C(\alpha, \beta, \delta_1, \delta_2)$ is some constant that depends only on $\alpha, \beta,$ $\delta_1, \delta_2$. The above inequality holds for all $(X_1, X_2) \ind Y$ if and only if $\alpha, \beta, \delta_1, \delta_2$ satisfy the following inequalities:
\begin{enumerate}
\item
$2\alpha + \beta = 2 +  \delta_1 + \delta_2$;
\item
$\beta \leq 1$;
\item
$\alpha \leq 1+\delta_1$, and $\alpha \leq 1+\delta_2$;
\item
$\alpha + \beta \leq 1+\delta_1+\delta_2$, which, combined with condition (1), is equivalent to $\alpha \geq 1$.
\end{enumerate}
\end{lemma}
\begin{proof}
We shall use Theorem \ref{thm: finite} to show this result. The above inequality is easily seen to be of the form in Theorem \ref{thm: EPI+BL}, where $A_1 = [1, 0, 1; 0, 1, 1]$, $A_2 = [1, 0, 0]$, $A_3 = [0, 1, 0]$, ${\mathbf r} = (2,1)$, $d_1 = \alpha$, and $d_2 = \beta$.
An exhaustive search of all possible subspaces $V$ that are in $\mathbf r$-product form where $\mathbf r = (2,1)$ is not hard to do. For simplicity, we refer to the axes in $\real^3$ as $X_1, X_2, Y$. Thus, the subspace $X_1$ is simply the subspace spanned by $(1, 0, 0)$.
\begin{enumerate}
\item
Equality (1) follows directly from equation
\eqref{eq: finite2} of Theorem \ref{thm: finite};
\item 
Inequality (2) follows from equation
\eqref{eq: finite1} of Theorem \ref{thm: finite}, by choosing $V = \phi \times Y$;
\item
Inequality (3) follows from equation
\eqref{eq: finite2} of Theorem \ref{thm: finite}, by choosing $V = X_1 \times \phi$ and $V = X_2 \times \phi$;
\item
Inequality (4) is obtained from equation
\eqref{eq: finite2} of Theorem \ref{thm: finite}, by a careful choice of $V = (X_1 + X_2) \times Y$, i.e. the subspace spanned by $(1, 1, 0)$ and $(0, 0, 1)$. 
\end{enumerate}
\end{proof}

\begin{claim}
For $\alpha, \beta < 1, \delta_1 = \delta_2 = \delta$ satisfying the conditions in Lemma \ref{lemma: abcd}, the following inequality holds:
\begin{align*}
&h(X_1+Y, X_2+Y) \geq (\alpha - \delta) h(X_1, X_2) + \beta h(Y) - \delta I(X_1; X_2) - D,
\end{align*}
where 
$$D = \frac{1}{2} \log \left( \frac{\beta^\beta(1-\beta)^{1-\beta}}{2^\beta}\left(1+\frac{\beta}{2\delta}\right)^{\alpha+\beta-1}\left(1-\frac{\beta}{2\delta}\right)^{\alpha-1} \right).$$
\end{claim}
\begin{proof}
For $\alpha, \beta, \delta_1, \delta_2$, the optimal constant $C$ is given by
\begin{align*}
e^{2C} &= \sup_{K_1, K_2, K_3, \rho} \frac{\left(\det \begin{pmatrix} &K_1 & \rho\sqrt{K_1K_2}\\ &\rho\sqrt{K_1K_2} &K_2 \end{pmatrix}\right)^\alpha \cdot K_3^\beta}{\det \begin{pmatrix} &K_1+K_3 & \rho\sqrt{K_1K_2}+K_3\\ &\rho\sqrt{K_1K_2}+K_3 &K_2+K_3 \end{pmatrix} K_1^{\delta_1} K_2^{\delta_2}}\\
&= \sup_{K_1, K_2, K_3, \rho} \frac{K_1^{\alpha-\delta_1}K_2^{\alpha-\delta_2}(1-\rho^2)^\alpha \cdot K_3^\beta}{K_1K_2(1-\rho^2) + K_3(K_1+K_2-2\rho \sqrt{K_1K_2})}.
\end{align*}
Calculating the above supremum for arbitrary $\alpha, \beta, \delta_1, \delta_2$ is cumbersome so we assume $\delta_1 = \delta_2 = \delta.$ The supremum simplifies to
\begin{align*}
e^{2C} &= \sup_{K_1, K_2, K_3, \rho} \frac{(K_1K_2)^{\alpha-\delta}(1-\rho^2)^\alpha \cdot K_3^\beta}{K_1K_2(1-\rho^2) + K_3(K_1+K_2-2\rho \sqrt{K_1K_2})}.
\end{align*}
For a fixed $K_1K_2$ and fixed $K_3$, it is clear that the optimal choice of $K_1 = K_2 = \sqrt {K_1K_2}$ maximizes the above expression. Thus, we assume that $K_1 = K_2 = K$ and obtain 
\begin{align*}
e^{2C} &= \sup_{K, K_3, \rho} \frac{(K)^{2\alpha-2\delta-1}(1-\rho^2)^\alpha \cdot K_3^\beta}{K(1-\rho^2) + 2K_3(1-\rho)}.
\end{align*}
Let $x \defn K_3/K$, and noting that $2\alpha-2\delta-1 = 1-\beta$, we obtain
\begin{align*}
e^{2C} &= \sup_{x \geq 0, \rho} \frac{x^\beta(1-\rho^2)^\alpha}{(1-\rho^2) + 2x(1-\rho)}\\
&= \sup_{x \geq 0, \rho} \frac{x^\beta(1-\rho)^\alpha(1+\rho)^\alpha}{(1-\rho)(1+\rho) + 2x(1-\rho)}\\
&=  \sup_{x \geq 0, \rho} \frac{x^\beta(1-\rho)^{\alpha-1}(1+\rho)^\alpha}{(1+\rho) + 2x}.
\end{align*}
For a fixed $\rho$, the maximum of the above expression is attained when 
\begin{align*}
x &= \frac{\beta(1+\rho)}{2(1-\beta)}.
\end{align*}
\ifvacomments 
Substituting this value of $x$, 
\begin{align*}
e^{2C} &= \sup_{\rho} \frac{(1+\rho)^\alpha(1-\rho)^{\alpha-1} \left(\frac{\beta(1+\rho)}{2(1-\beta)}\right)^\beta}{\frac{\beta(1+\rho)}{(1-\beta)}+ (1+\rho)}\\
&= \sup_{\rho} \frac{(1+\rho)^\alpha(1-\rho)^{\alpha-1} \left(\frac{\beta(1+\rho)}{2(1-\beta)}\right)^\beta}{\frac{1+\rho}{1-\beta}}\\
&=  \frac{\beta^\beta(1-\beta)^{1-\beta}}{2^\beta}\sup_{\rho}  (1+\rho)^{\alpha+\beta-1}(1-\rho)^{\alpha-1}.
\end{align*}
Differentiating with respect to $\rho$, the supremum is seen to be attained when $\rho = \frac{\beta}{2\alpha+\beta-2} = \frac{\beta}{2\delta}.$ 
Substituting this, we get
\begin{align*}
e^{2C} &=  \frac{\beta^\beta(1-\beta)^{1-\beta}}{2^\beta}\left(1+\frac{\beta}{2\delta}\right)^{\alpha+\beta-1}\left(1-\frac{\beta}{2\delta}\right)^{\alpha-1}.
\end{align*}
This leads to the entropy inequality
\begin{align*}
&h(X_1+Y, X_2+Y) \\
&\geq \alpha h(X_1, X_2) + \beta h(Y) - \delta h(X_1) - \delta h(X_2) - \frac{1}{2} \log \left( \frac{\beta^\beta(1-\beta)^{1-\beta}}{2^\beta}\left(1+\frac{\beta}{2\delta}\right)^{\alpha+\beta-1}\left(1-\frac{\beta}{2\delta}\right)^{\alpha-1} \right)\\
&= (\alpha - \delta) h(X_1, X_2) + \beta h(Y) - \delta I(X_1; X_2) - \frac{1}{2} \log \left( \frac{\beta^\beta(1-\beta)^{1-\beta}}{2^\beta}\left(1+\frac{\beta}{2\delta}\right)^{\alpha+\beta-1}\left(1-\frac{\beta}{2\delta}\right)^{\alpha-1} \right).
\end{align*}
Notice that the mutual information term $I(X_1; X_2)$ accounts for the dependency between $X_1$ and $X_2$.\end{proof}

\section{Conclusion} \label{section: end}

In this paper, we established a new inequality that unifies the BLI and the EPI by establishing subadditivity of certain entropic functionals. There are several interesting research directions that are worth pursuing. We did not address the questions of extremizability and uniqueness of extremizers in this work. One reason for this is that Theorem~\ref{thm: EPI+BL} is established by taking the limit as $\epsilon$ and $\delta$ go to 0. When $\epsilon$ and $\delta$ are strictly bounded away from 0, the extremizer of $s_{\epsilon, \delta}(\cdot)$ under a covariance constraint exists and is a unique Gaussian distribution. However, these existence and uniqueness properties need not hold in the limit as $\epsilon, \delta \to 0$. In general, such a proof strategy is a powerful tool for proving inequalities, but may not always succeed in identifying necessary and sufficient conditions for equality. For this reason, alternate proof strategies that rely on heat flow based arguments \cite{BarEtal06, BenEtal08, CarEtal09} or optimal transport methods \cite{Bar98, Rio17} are worth exploring as well. After a preprint of this work appeared online, an optimal transport-based proof of Theorem~\ref{thm: EPI+BL} was discovered in Courtade~\cite{Cou19}. Shortly thereafter, Courtade and Liu~\cite{CouLiu20} proved Theorem~\ref{thm: EPI+BL} as a limiting case of the forward-reverse Brascamp-Lieb inequality~\cite{LiuEtal18} and gave an alternate proof of Theorem~\ref{thm: finite}.  

Finally, although our results generalize the BLI and the EPI to vector random variables with more general independence properties, these independence properties are still quite restrictive. For instance, the inequalities we derived do not encompass the monotonicity of entropy power family of results \cite{abbn04,mab07,mag18}. It would be interesting to generalize our inequalities to include the above family as well. Another (related) direction to pursue would be to establish similar entropy inequalities under weaker independence conditions. 

\section*{Acknowledgements}

The research of VA was
supported by the NSF grants
CNS-1527846, CCF-1618145, CCF-1901004, CIF-2007965, the NSF Science \& Technology
Center grant CCF-0939370 (Science of Information), and the
William and Flora Hewlett Foundation supported Center for
Long Term Cybersecurity at Berkeley. VJ acknowledges support from NSF grants CCF-1841190 and CCF-1907786, and is grateful to the Department of Information Engineering at CUHK for hosting him in July 2018, when a part of this work was done. The research of CN was supported by GRF grants 14303714, 14231916, 14206518 and a discretionary fund of the Vice Chancellor of CUHK.
 
\bibliographystyle{unsrt}
\bibliography{refs.bib}

\begin{appendix}
\section{Supporting results for Theorem \ref{thm: EPI+BL}}
\begin{lemma}\label{lemma: markov_ind}
Let $X, Y,$ and  $Z$ be random variables taking values in $\real^{n_X}, \real^{n_Y},$ and $\real^{n_Z}$ respectively, such that the following hold: (a) $(X, Y, Z)$ has a strictly positive density on $\real^{n_X+n_Y+n_Z}$; (b) $X \rt Y \rt Z$; and (c) $X \rt Z \rt Y$. Then $X \ind (Y,Z)$.
\end{lemma}
\begin{proof}
For any $x \in \real^{n_X}$, $y \in \real^{n_Y}$, and $z \in \real^{n_Z}$, we have that
\begin{align}\label{eq: entire}
p_{X|YZ}(x|y,z) = p_{X|Y}(x|y) = p_{X|Z}(x|z),
\end{align}
where we used the assumed strict positivity of the density of $(X,Y,Z)$ to write the above equations. Fix $y_0 \in \real^{n_Y}$. For any $z \in \real^{n_Z}$, we have 
\begin{align*}
p_{X|Z}(x|z) = p_{X|Y}(x | y_0).
\end{align*}
Integrating both sides of the above equality with respect to $p_Z(z)$, we obtain
\begin{align*}
p_X(x) = p_{X|Y}(x|y_0).
\end{align*}
Since $y_0$ was chosen arbitrarily, we conclude that $X \ind Y$. A similar argument shows that $X \ind Z$. Using equation \eqref{eq: entire}, we conclude that $X \ind (Y,Z)$.
\end{proof}

\begin{lemma}\label{lemma: xz_ind}
Let $X_1$ and $X_2$ be $\real^n$-valued random variables and let $(Z_1, Z_2) \ind (X_1, X_2)$ be such that $(Z_1, Z_2) \sim \cN(0, I_{2n \times 2n})$. If $(X_1 +Z_1) \ind (X_2+Z_2)$, then $X_1 \ind X_2$.
\end{lemma}
\begin{proof}
Using the independence of $(X_1+Z_1)$ and $(X_2+Z_2)$, we have that for any $t_1, t_2 \in \real^n$, 
\begin{align}
\phi_{X_1+Z_1, X_2+Z_2}(t_1, t_2) &:=  \E e^{i\langle t_1, X_1+Z_1\rangle  + i\langle t_2, X_2+Z_2 \rangle}\\
&= 
\E e^{i\langle t_1, X_1+Z_1\rangle} 
\E e^{i\langle t_2, X_2+Z_2 \rangle } 
\\
&= \E e^{i\langle t_1, X_1\rangle}
 \E e^{i\langle t_2, X_2 \rangle }
\E e^{i\langle t_1, Z_1\rangle}
\E e^{i\langle t_2, Z_2 \rangle } \\
&= \phi_{X_1}(t_1) 
\phi_{X_2}(t_2) 
\phi_{Z_1, Z_2}(t_1, t_2).
\end{align}
However, using the independence $(X_1, X_2) \ind (Z_1, Z_2)$, we also have
\begin{align}
\phi_{X_1+Z_1, X_2+Z_2}(t_1, t_2) &=  \E e^{i\langle t_1, X_1+Z_1\rangle  + i\langle t_2, X_2+Z_2 \rangle}\\
&= \E e^{i\langle t_1, X_1\rangle + i \langle t_2, X_2\rangle} 
 \E e^{i\langle t_1, Z_1\rangle + i \langle t_2, Z_2\rangle}\\
&= \phi_{X_1, X_2}(t_1, t_2) 
\phi_{Z_1, Z_2}(t_1, t_2).
\end{align}
Since $\phi_{Z_1, Z_2}(\cdot, \cdot)$ has no zeros ($Z_i$'s being independent standard Gaussian random variables), we conclude that
\begin{equation}
\phi_{X_1, X_2}(t_1, t_2)  =  \phi_{X_1}(t_1) 
\phi_{X_2}(t_2),
\end{equation}
that is, $X_1 \ind X_2$.
\end{proof}

\begin{lemma}\label{lemma: heat}
Let $X$ be an $\real^n$-valued random variable with density $p_X(x)$ and $Z \sim \cN(0, I_{n \times n})$ be independent of $X$. Suppose that   $\E[\Psi(X)] < \infty$ for some nonnegative continuous function $\Psi: \real^n \mapsto \real$, satisfying $\int_{\real^n} e^{-\Psi(x)} dx < \infty$ and $\lim_{\delta \to 0} \E[\Psi(X + \sqrt{\delta} Z)] = \E[\Psi(X)]$.  (Note that, for instance, $\Psi(X) = \|X\|_p, p \geq 1$ satisfies the conditions.) 
Then the following equality holds:
\begin{equation}
    \lim_{\delta \to 0} h(X + \sqrt \delta Z) = h(X).
\end{equation}
\end{lemma}
\begin{proof}
Our proof relies on the following (lower semi-continuity) result from Posner \cite[Theorem 1]{Pos75}: If $P_m, Q_m$ are Borel probability distributions on a Polish space with $P_m \stackrel{w}{\Rightarrow} P$ and $Q_m \stackrel{w}{\Rightarrow} Q$, then
$$  D(P\|Q) \leq \liminf_m D(P_m\|Q_m), $$
where $D(P\|Q)$ denotes the relative entropy of the distribution $P$ with respect to the distribution $Q$. Picking an arbitrary sequence $\{\delta_m\}_{m \geq 1}$ that converges to $0$, let $X_m = X + \sqrt \delta_m Z.$ Using characteristic functions (or otherwise), it is easy to check that $X_m$ converges to $X$ in distribution.
Let $P_m$ denote the distribution  of $X_m$ and $P$ denote the distribution of $X$. Let $Q_m=Q$ be the distribution corresponding to the density function $Ce^{-\Psi(x)}$. Note that
$$D(P_m\|Q) = \E[\Psi(X+ \sqrt{\delta_m}Z)] - h(X + \sqrt{\delta_m}Z) - \log C. $$
Therefore, we have
\begin{align*}
 & \E[\Psi(X)] - h(X) - \log C \\
 & = D(P\|Q) \stackrel{(a)}{\leq} \liminf_m D(P_m\|Q) \\
 & \quad \leq \liminf_m \Big\{\E[\Psi(X + \sqrt{\delta_m}Z)] - h(X + \sqrt{\delta_m}Z) - \log C\Big\}.\\
 & \quad \stackrel{(b)}{=} \E[\Psi(X)]  - \limsup_m h(X + \sqrt{\delta_m}Z) - \log C.
\end{align*}
Here $(a)$ follows from the Posner's result and $(b)$ follows from  assumption (2). Hence
\begin{equation}\label{eq: limsup_XZ}
    \limsup_{m \to \infty} h(X + \sqrt{\delta_m}Z) \leq h(X).
\end{equation}
On the other hand, non-negativity of mutual information, $I(Z;X + \sqrt{\delta_m} Z) \geq 0$, yields $ h(X+\sqrt \delta_m Z) \geq h(X).$ Taking the $\liminf$ on both sides of this equality, we conclude
\begin{equation}\label{eq: liminf_XZ}
    \liminf_{m \to \infty} h(X + \sqrt \delta_m Z) \geq h(X).
\end{equation}
Inequalities \eqref{eq: limsup_XZ} and \eqref{eq: liminf_XZ} yield the equality 
\begin{equation}
    \lim_{m \to \infty} h(X + \sqrt \delta_m Z) = h(X),
\end{equation}
and concludes the proof.
\end{proof}

\section{Supporting results for Claim \ref{claim: sufficient_finite}}

\begin{lemma}\label{lemma:ball}
Given subspaces $K_i \subseteq \mbbR^{r_i}$ for $1 \le i \le k$, with $s_i := \mdim(K_i)$,
let $K := K_1 \times \ldots \times K_k$ denote the corresponding $\mbfr$-product subspace of 
$\mbbR^n$, where $n := \sum_{i=1}^k r_i$. Let $A_1 = \left[ A_{11} \ldots A_{1k} \right]$,
with $A_{1i}$ an $n_1 \times r_i$ matrix of rank $r_i$ for $1 \le i \le k$ as above. 
Then there is some $\eta > 0$ such that for all choices of $(K_i, 1 \le i \le k)$
where at least one $s_i$ is strictly positive, for all unit vectors $x \in A_1 K$
(i.e. $x^T x = 1$), there exists some unit vector $v_i \in A_{1i} K_i$ for some 
$1 \le i \le k$ such that $|x^T v_i| \ge \eta$.
\end{lemma}
\begin{proof}
Suppose to the contrary that we can find a sequence
$( (x(t), (K_1(t) \ldots, K_k(t))), t \ge 1)$ of unit vectors and subspaces that violates the 
condition, i.e. such that
\[
\lim_{t \to \infty} \sup_{1 \le i \le k} \sup_{v_i \in A_i K_i(t) : v_i^T v_i = 1} |x(t)^T v_i| = 0.
\]
By going to a subsequence if necessary we can assume that there exist some choices of $1 \le s_i \le r_i$ for
$1 \le i \le k$ with at least one of the $s_i$ being strictly positive, such that we have
$\mdim(K_i(t)) = s_i$ for all $t \ge 1$. Since the space of all $s_i$-dimensional subspaces of 
$\mbbR^{r_i}$ is compact in the usual topology (i.e. as the corresponding Grassmanian), by going
to a further subsequence if necessary we can assume that each $K_i(t)$ converges to a limit $K_i$ as $t \to \infty$,
where $\mdim(K_i) = s_i$. Since the set of unit vectors in $\mbbR^{n_1}$ is compact, by going to a further
subsequence if necessary we can assume that $x(t)$ converges to a unit vector $x \in \mbbR^{n_1}$ as 
$t \to \infty$. Since we have $x(t) \in A_1 K(t)$ for all $t \ge 1$ (where 
$K(t) := K_1(t) \times \ldots \times K_k(t)$), we must have $x \in A_1 K$ 
(where $K := K_1 \times \ldots \times K_k$). We thus have 
$x^T v_i = 0$ for all unit vectors $v_i \in A_1 K_i$ for all $1 \le i \le k$. But this is a contradiction,
because $x$ is itself in the linear span of such vectors.  
\end{proof}
~\\

\begin{corollary}       \label{corollary:ball}
There is positive constant $\delta^2 > 0$ depending only on 
$W_1, \ldots, W_k$ (and in particular not depending on the 
$(\hatSigma_i, 1 \le i \le k)$ or the choices of the bases
$\{b_{i1}, b_{i2}, \ldots, b_{i r_i} \}$ for $1 \le i \le k$)
such that, for all $1 \le l \le n'$, we have
\[
M_l \succeq \delta^2 C_l,
\]
where $C_l$ is a positive semidefinite matrix all of whose eigenvalues are either $0$ or $1$
and where the subspace corresponding to $C_l$ is $A_1 V_l$.
\end{corollary}
\begin{proof}
Let $\eta > 0$ be as in the Lemma. For each unit vector $x \in A_1 V_l$ there exists some $1 \le i \le k$
and a unit vector $v_i \in A_1 V_{il}$ such that $|x^T v_i| \ge \eta$. 
Since $\{\tilb_{i1}, \ldots, \tilb_{i s_{il}}\}$ is an orthonormal basis for 
$A_1 V_{il}$, This means means that there is some $1 \le u_i \le s_{il}$ such that
$|x^T \tilb_{i u_i}| \ge \delta$, where we define $\delta := \frac{1}{n} \eta$ 
and we have used $s_{il} \le r_i \le n$. 
Recalling that $M_l := \sum_{i=1}^k \sum_{u_i =1}^{s_{il}} \tilb_{i u_i} \tilb_{i u_i}^T$,
it follows that
\[
x^T M_l x_l \ge \delta^2.
\]
Since this holds for all unit vectors $x \in A_1 V_l$, this proves the corollary. 
\end{proof}
\color{black}
\end{appendix}

\end{document}